\documentclass{article}
\usepackage[latin1]{inputenc}
\usepackage[dvips]{epsfig}
\usepackage{graphicx}
\usepackage{color}
\usepackage{latexsym}
\usepackage{verbatim}
\usepackage{amsmath}
\usepackage{amsthm}
\usepackage{amssymb}
\usepackage{subfigure}
\usepackage{oldgerm}
\usepackage{rotating}
\usepackage{fancyhdr}
\usepackage{psfrag}
\usepackage[subfigure]{ccaption}
\usepackage{a4wide}

\DeclareMathAlphabet{\mathpzc}{OT1}{pzc}{m}{it}

\newcommand{\C}{\mathbb{C}}
\newcommand{\D}{\mathbb{D}}
\newcommand{\Hilb}{\mathbb{H}}
\newcommand{\M}{\mathbb{M}}
\newcommand{\N}{\mathbb{N}}
\newcommand{\R}{\mathbb{R}}

\newcommand{\orthoIdomtN}{\widetilde{\D}_t^N}
\newcommand{\IdomtN}{\D_t^N}
\newcommand{\Idomt}{\D_t}

\newcommand{\Mdomt}{\M_t}
\newcommand{\dom}{X}            
\newcommand{\tdom}{]0,t_1[}            

\newcommand{\sph}{Y}
\newcommand{\sphere}{{S^2}}

 \DeclareMathOperator{\id}{Id}
 \DeclareMathOperator{\range}{Range}

\newcommand{\operatordef}[1]{{\mathcal #1}}
\newcommand{\mnt}{\textswab{m}}
\newcommand{\oA}{\operatordef{A}}

\newcommand{\oC}{\operatordef{C}}
\newcommand{\oD}{\operatordef{D}}
\newcommand{\oE}{\operatordef{E}}

\newcommand{\oK}{\operatordef{K}}
\newcommand{\oL}{\operatordef{L}}
\newcommand{\oM}{\operatordef{M}}

\newcommand{\oP}{\operatordef{P}}
\newcommand{\oQ}{\operatordef{Q}}

\newcommand{\oS}{\operatordef{S}}

\newcommand{\vectordef}[1]{{\mathbf #1}}

\newcommand{\vI}{\vectordef{I}}

\newcommand{\vQ}{\vectordef{Q}}

\newcommand{\vY}{\vectordef{Y}}

\newcommand{\ve}{\vectordef{e}}

\newcommand{\scc}{\sigma}        
\newcommand{\abc}{\kappa}        
\newcommand{\I}{I}               
\newcommand{\B}{B}               
\newcommand{\Q}{Q}               
\newcommand{\vIN}{\vI_N}          
\newcommand{\devIN}{\tilde \I_N}               

\theoremstyle{plain}

\newtheorem{thm}{\bf Theorem}[section]
\newtheorem{lem}[thm]{\bf Lemma}
\newtheorem{prop}[thm]{\bf Proposition}

\newtheorem{defs}[thm]{\bf Definition}

\newtheorem*{prob*}{\bf Problem}

\theoremstyle{remark}
\newtheorem{rem}[thm]{\bf  Remark}

\numberwithin{equation}{section}

\newcommand{\secref}[1]{Section~\ref{#1}}
\newcommand{\apxref}[1]{Appendix~\ref{#1}}

\newcommand{\lemref}[1]{Lemma~\ref{#1}}
\newcommand{\defnref}[1]{Definition~\ref{#1}}

\newcommand{\figref}[1]{Figure~\ref{#1}}

\makeatletter
\def\qdots{\mathinner{\mkern1mu\raise\p@
\vbox{\kern7\p@\hbox{.}}\mkern2mu
\raise4\p@\hbox{.}\mkern2mu\raise7\p@\hbox{.}\mkern1mu}}
\makeatother

\newcommand{\includegraphicschoice}[2]{\includegraphics[#1]{#2}}

\begin{document}

\title{Diffusive Corrections to $P_N$ Approximations}
\author{Matthias Sch\"afer\footnote{Fraunhofer-Institut f\"ur 
        Techno- und Wirtschaftsmathematik, Fraunhofer-Platz 1, 
        67663 Kaiserslautern, Germany, 
        {\tt matthias.schaefer@itwm.fhg.de}} 
        \and Martin Frank\footnote{Department of Mathematics, 
        University of Kaiserslautern, Erwin-Schr\"odinger-Strasse, 
        67663 Kaiserslautern, Germany, 
        {\tt frank@mathematik.uni-kl.de}}
        \and C. David Levermore\footnote{Department of Mathematics
        {\em and} Institute for Physical Science and Technology,
        University of Maryland, College Park, MD 20742, USA,
        {\tt lvrmr@math.umd.edu}}}

\maketitle
\abstract{In this paper, we investigate moment methods from a general 
point of view using an operator notation.  This theoretical approach 
lets us explore the moment closure problem in more detail.  This gives 
rise to a new idea, proposed in \cite{Levermore2005, Levermore2009}, 
of how to improve the well-known $P_N$ approximations.  We 
systematically develop a diffusive correction to the $P_N$ equations 
from the operator formulation --- the so-called $D_N$ approximation.  
We validate the new approach with numerical examples in one and two 
dimensions.}

\section{Introduction}
%
Developing simplified methods for the simulation of radiative transfer
requires taking into account the physical situation that will be 
analyzed.  There are two important limits: optically thick and 
optically thin media.  In optically thin media there are very few 
particles that interact with the radiation.  The distances that 
photons would typically travel before they are scattered or absorbed 
are therefore very long compared to the domain size.  On the other 
hand, in optically thick regimes those distances are very short 
compared to the domain size.
\begin{figure}[h!]
    \centering
    \subbottom[Optically thin medium.]
        {\label{cha_model:fig:optical_thin_medium}
        \epsfig{file=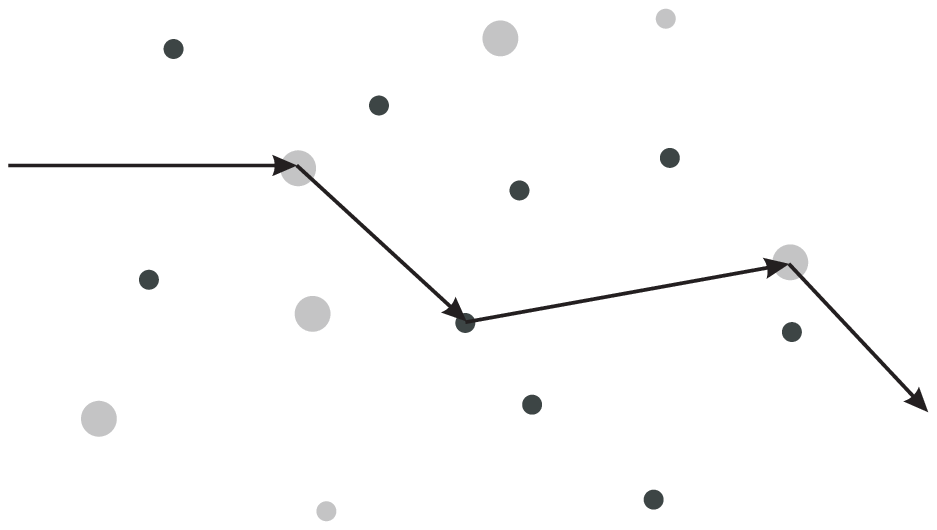,width=0.425\linewidth}}
    \subbottom[Optically thick medium.]
        {\label{cha_model:fig:optical_thick_medium}
        \epsfig{file=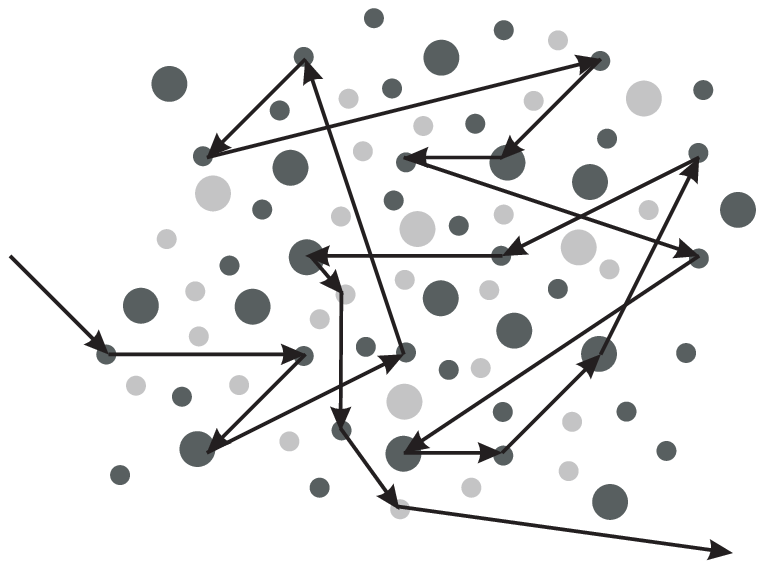,width=0.425\linewidth}}
    \caption{Path of a photon in optically thin and optically thick media.}
    \label{cha_model:fig:optical_thick_and_thin_media}
\end{figure}

   The different regimes can be characterized by the mean free paths 
$\chi$ for scattering and $\varepsilon$ for absorption. 
Large mean free paths represent an optically thin regime, small mean
free paths represent optically thick regimes.  One problem in this
characterization is, that many materials are optically thick in a
specific frequency range and optically thin in other ranges.
Additionally, there is a transition regime between the two.  This is 
the region of mean free paths between optically thin and thick media. 
In this regime, photons travel further than in optically thick 
situations, but not far enough as that the regime counts as optically 
thin.

   An example of an optically thin regime is radiation propagating in
vacuum.  Optically thick regimes can be found in glass cooling 
processes or combustion chambers.  There are also situations where all 
these regimes play a role.  For instance, during the reentry of a 
space craft into the atmosphere the regime goes from optically thin
(space) through a transition (higher atmosphere) into optically thick 
(lower atmosphere).

   Due to large mean free paths in optically thin regimes it is 
possible to track the trace of many single photons until they leave 
the computational domain or undergo absorption.  This is the basic 
idea of ray tracing and Monte Carlo methods.  If there are very few
scattering events, following the path of one photon is rather simple.  
These particle methods are often used in astrophysics where the 
physical conditions are usually in a way that these methods can be 
applied successfully.  For examples see 
\cite{Pascucci2004, Steinacker2006, Wolf2001, Wolf1999}.

   In optically thick regimes, following the path of a single photon 
is almost impossible because it undergoes too many scattering events
before it reaches its destination (leaving the computational domain
or being absorbed).  Therefore, other methods are used.  They are 
usually based on the diffusion approximation \cite{Rosseland}, e.g.\ 
the Simplified $P_N$ approximation 
\cite{Gel61, LarThoKlaSeaGot02, Pom93, Thoemmes} or flux-limited 
diffusion \cite{Lev84}.

   The transition regime is the region of optical depth that is 
located between optically thin and optically thick regimes.  Methods
that work well in optically thin media are computationally too
expensive in these regimes.  Methods that work well in optically thick 
media, on the other hand, give poor results for low order
approximations and have high computational costs if one increases
the order.  Therefore, in these regimes new methods have to be 
developed.

   These new approximations have to recover traditional reduced models
for small mean free paths.  Further, in the transition regime, they
have to be more accurate than the simplified models and should be
solvable more efficiently than the full kinetic approaches.

   The field of use for transition regime models can be found in
stand-alone solvers for problems that lie completely within the
optically thick and transition regime.  For problems where the order
of the mean free path also covers the optically thin regime, the new
approaches could be used in hybrid methods.

   In this work, our starting point will be moment methods.  These are 
usually derived based on the assumption that in highly scattering 
materials the photon distribution is driven toward a local equilibrium
Therefore, the radiative intensity distribution is almost isotropic at 
every point.  If this is the case, instead of treating the full 
intensity distribution, we can restrict our analysis to quantities 
that are averages of the directional distribution function over all
directions.  These quantities are e.g.\ the spectral energy
distribution, the spectral flux and the spectral pressure.  Averages 
of products of intensity distribution with directional test functions 
are called moments of the intensity function.  Often, one is only
interested in these averaged quantities.

   There is a variety of moment closures.  One of the first was the 
$P_N$ closure, which was developed by Chandrasekhar 
\cite{Chandrasekhar1944}.  Other approaches include the minimum
entropy closure \cite{AniPenSam91,DubFeu99,Min78,Lev84,Lev96}.
A recent approach applies methods from the study of dynamical systems
to moment closure \cite{FraSei09,SeiFra09}.

   In \secref{sec:devdec} we introduce the main concept of moment 
methods in an operator notation.  This theoretical approach lets us 
explore the moment closure problem in more detail.  This 
investigation gives rise to a new approach, proposed in 
\cite{Levermore2005, Levermore2009}, of how to improve the 
well-known $P_N$ approximations (\secref{sec:devapp}). 
Using the operator formulation, we systematically develop a diffusive 
correction to the $P_N$ equations --- the so-called $D_N$ 
approximation.  We do not treat boundaries here; they are considered
in \cite{FrankHauckLevermore2009}.  The partial differential 
equations behind the operator notation are developed for a simple 
example in \secref{cha_moment:sec:simplified_example_problem_1d}, 
before we develop the general $P_N$ and $D_N$ equations in 
Sections \ref{cha_moment:sec:analysis_operators_pn_approximation} and 
\ref{cha_moment:sec:analysis_operators_modified_pn_approximation}
respectively.  Numerical examples in one and two dimensions are then
shown in \secref{cha_numerics:sec:numerical_results}.  We find that
the solution of the $D_N$ equations is at least as accurate as the
solution of the $P_N$ equation of two orders higher. 

\section{Moment Models and Deviation Decomposition}
\label{sec:devdec}

   We consider radiation in a spatial domain $X$ with boundary 
$\partial X$ whose intensity $I(t,x,\Omega)$ at time $t\geq0$, 
position $x\in X$, and direction $\Omega\in S^2$ is governed by the 
frequency-averaged radiative transfer equation (RTE) 
\begin{multline}
    \label{cha_phy:eq:total_radiative_transfer_equation}
    \frac{1}{c} \partial_t \I(t,x,\Omega) 
    + \Omega \cdot \nabla_x \I(t,x,\Omega)
    + (\scc(x) + \abc(x)) \I(t,x,\Omega) 
\\
    = \frac{\scc(x)}{4\pi}
      \int_\sphere \Phi(x,\Omega \cdot \Omega') 
                   \I(t,x,\Omega') \, d\Omega' 
      + \abc(x) \B(T(x)) + \Q(t,x,\Omega) \,.
\end{multline}
Here $\scc(x)$ and $\abc(x)$ are the scattering and absorption 
coefficients, $\Phi(x,\Omega \cdot \Omega')>0$ is the scattering 
redistribution function, $\B(T(x))>0$ is the blackbody emission 
intensity at temperature $T(x)>0$, and $\Q(t,x,\Omega)$ is the
emission due to other sources.  The scattering redistribution 
function satisfies the normalization
\begin{equation}
  \label{cha_phy:eq:redistribution_normalization}
  \int_\sphere \Phi(x,\Omega \cdot \Omega') \, d\Omega = 1 \,. 
\end{equation}
The fact that $\scc(x)$, $\abc(x)$, and $\B(T(x))$ are 
independent of $\Omega$, while $\Phi(x,\Omega \cdot \Omega')$
depends on $\Omega \cdot \Omega'$ is consistent with a stationary, 
isotropic background medium.  The fact that these functions are 
independent of $t$ means that the heat capacity of this medium is
large.   

   We will express the different parts of equation
(\ref{cha_phy:eq:total_radiative_transfer_equation}) in terms of 
operators.

\begin{defs}
    \label{cha_model:def:operators_for_RTE}
    We define the operators
\begin{subequations}
    \label{cha_model:eq:operator_definitions_RTE}
    \begin{align}
        (\oA \I)(t,x,\Omega) 
        & = \Omega \cdot \nabla_x \I(t,x,\Omega) \,,
\\
        (\oS \I)(t,x,\Omega) 
        & = \frac{\scc(x)}{4\pi}
            \int_{S^2} \Phi(x,\Omega \cdot\Omega') \,
                       \I(t,x,\Omega') \, d\Omega' \,,
\\
        (\oK \I)(t,x,\Omega) 
        & = (\abc(x) + \scc(x)) \I(t,x,\Omega) 
            - (\oS \I)(t,x,\Omega) \,,
\\
        (\oL \I)(t,x,\Omega) 
        & = (\oA + \oK) \I(t,x,\Omega) \,,
\\
        \oQ(t,x,\Omega) 
        & = \abc(x) \B(T) + \Q(t,x,\Omega) \,.
        \end{align}
    \end{subequations}
\end{defs}
Here $\oA$ is advection, $\oS$ is scattering, $\oK$ is
total interaction due to scattering and absorption, and
$\oQ(t,x,\Omega)$ is total emission.  Using these operators 
the RTE reads
\begin{equation}
   \label{cha_model:eq:operator_version_RTE}
   \frac{1}{c} \partial_t \I(t,x,\Omega) + \oL \I(t,x,\Omega) 
   = \oQ(t,x,\Omega) \,.
\end{equation}
We impose homogeneous boundary conditions.  Let
\begin{equation}
    \Gamma = \partial \dom \times S^2 \,, 
    \quad \text{and} \quad
    \Gamma^\pm = \{ (x,\Omega) \in \Gamma \;:\; 
                    \pm n(x) \cdot \Omega > 0 \} \,,
\end{equation}
where $n$ is the outward unit normal vector.  Appropriate boundary 
conditions are
\begin{equation}
    \I(t,x,\Omega) = 0 \qquad 
    \forall t>0 \; \text{and} \; (x,\Omega) \in \Gamma^- \;.
\end{equation}
Together with the initial condition
\begin{equation}
    \I(0,x,\Omega) = I_0(x,\Omega) \,,
\end{equation}
we have a well-posed problem.  Under certain (physically reasonable) 
assumptions on the scattering and absorption coefficients $\scc$ and 
$\abc$, and for $\oQ \in L^2(\tdom \times \dom \times \sphere,\R)$, 
$\oQ(t,x,\Omega)\geq 0$ there exists a unique solution 
$\I \in \{ \I \in \Idomt \; : \; \I=0 \; \text{on} \; \Gamma^-\}$ 
(cf.~\cite{DautrayLions1993}), where
\begin{equation}
    \Idomt \subset L^2(\tdom \times \dom \times \sphere, \R).
\end{equation}
Moment methods have a long history \cite{Chandrasekhar1944,Davison}.
Nevertheless they are still used for solving radiative or neutron 
transfer problems in situations where computational time is of concern.  
The main idea of moment methods is to derive an approximation to the 
radiative intensity distribution with respect to its directional 
moments.  This relation is used to find an expression for the closure 
relation.  There are several ways to find such approximations.  The $P_N$ 
approach expresses the radiative intensity distribution as a series
expansion of spherical harmonics.  The minimum entropy approaches use
an expression with an exponential function.  But, independent of how 
these methods approximate the intensity, all of them have in common 
that the unknown coefficients in their expansion are somehow related 
to the directional moments of the intensity function.

   Moments are directional averages of the intensity distribution
multiplied with a test function that depends on the direction 
$\Omega$.  As test functions one can choose between several
possibilities.  We will use spherical harmonics, denoted by 
$\sph^k_l$, cf.\ Appendix \ref{apx:A}.
There are several reasons for choosing these functions.  First, they
form an orthonormal basis of the function space $\Hilb^m(\sphere,\C)$
 and therefore, after transformation also of $L^2(S^2, \R)$.  Hence, 
they can be used for a complete description of the dependence of the 
intensity distribution on the direction $\Omega$.  In other words, 
each intensity distribution can be represented by the series expansion
\begin{equation}
    \label{cha_moment:eq:infinite_series_expansion}
    \I(t,x,\Omega) 
    = \sum_{l=0}^\infty \sum_{k=-l}^l \vI^k_l(t,x) \sph^k_l(\Omega) \;,
\end{equation}
with moments
\begin{equation}
    \label{cha_moment:eq:definition_moments}
    \vI^k_l(t,x) = \int_{\sphere} \overline{\sph^k_l(\Omega)} 
                                  \I(t,x,\Omega) d\Omega \;.
\end{equation}
To deal with spherical harmonics and the related moments we define
\begin{defs}
    Moments of the radiative intensity are generated by the operator
    \begin{equation}
        \begin{split}
            \mnt^k_l \; : \; & \Idomt \to L^2(\tdom \times \dom, \C) \\
                             & \I(t,x,\Omega) \mapsto \int_{\sphere} \overline{\sph^k_l(\Omega)}
                               \I(t,x,\Omega) d\Omega \;,
        \end{split}
    \end{equation}
    and we write for the moment of order $l$ and degree $k$
    \begin{equation}
        \vI^k_l(t,x) = \mnt^k_l\left( \I(t,x,\Omega) \right) \;.
    \end{equation}
\end{defs}

   We allow the moments to be complex valued.  A real-valued 
approximation to the radiative intensity is obtained by taking the 
real part of these equations. The moments depend on time and space.  
For convenience we neglect this in the notation if it is clear to what 
we are referring and write $\vI^k_l = \vI^k_l(t,x)$.

   A vector $\vI$ of all moments belongs to
\begin{equation}
    \Mdomt = \left\{\vI = (\ldots,\vI^k_l,\ldots)^T \; : \; l\in \N_0, \; k\in\{-l,\ldots,l\}
    \right\} \subseteq l^2\left( L^2(\tdom \times \dom, \C) \right),
\end{equation}
where $l^2$ denotes the space of all square summable sequences.  We
introduce

\begin{defs}
    \label{cha_moment:def:moment_operators}
    We define the 
    ``Intensity to Moment'' operator
    \begin{equation}
        \begin{split}
            \oM \; : \; & \Idomt \to \Mdomt \\
                        & \I(t,x,\Omega) \mapsto \vI(t,x) \;.
        \end{split}
    \end{equation}
    The inverse transformation is given by the ``Moment to
    Intensity'' or ``Expansion'' operator
    \begin{equation}
        \begin{split}
            \oE \; : \; & \Mdomt \to \Idomt \\
                             & \vI(t,x) \mapsto \sum_{l=0}^\infty \sum_{k=-l}^l \vI^k_l(t,x)
                             \sph^k_l(\Omega)\;.
        \end{split}
    \end{equation}
\end{defs}
By their construction the operators $\oM$ and $\oE$ are linear,
bounded and continuous.  Furthermore, it is easy to see that both 
operators are bijective.

   So far we have replaced the unknown dependence in $\Omega$ by 
infinitely many unknown moments.  This does not help us to solve the 
RTE.  A usual approach to overcome this problem is to assume that 
finitely many moments are sufficient to describe the intensity 
function.   This reduces the amount of unknowns to a finite number 
and the problem can be handled much more easily.  Assuming that only 
the moments up to order $N$ are relevant gives an approximation 
$\I_N(t,x,\Omega)$ to $\I(t,x,\Omega)$
\begin{equation}
    \label{cha_moment:eq:RTE_approx_by_moment_upto_order_N}
    \I(t,x,\Omega) \approx \I_N(t,x,\Omega) 
                   = \sum_{l=0}^N \sum_{k=-l}^l \vI^k_l(t,x) \sph^k_l(\Omega)\;.
\end{equation}
and it holds
\begin{equation}
    \lim_{N\to\infty} \I_N(t,x,\Omega) = \I(t,x,\Omega).
\end{equation}
The finite set of moments can be represented by the vector
\begin{equation}
    \vIN = (\vI^{0}_0, \vI^{-1}_0,\vI^{0}_1,\ldots,\vI^{N-1}_N,\vI^{N}_N)^T
\end{equation}
and we define the set of restricted vectors of moments as
\begin{equation}
    \Mdomt^N =\left\{\vI_N \in \left(L^2(\tdom \times \dom, \C)\right)^{(N+1)^2}\right\} \;.
\end{equation}
Note that $\Mdomt^N$ is isomorphic to a subspace of $\Mdomt$.

   We restrict ourself to approximations of the radiative intensity of
odd orders.  There are several reasons for this.  First of all, even 
order approximations do not contain more information than odd order 
approaches.  Therefore, they only introduce more moments and are 
computationally more expensive without giving any advantage.  A second 
point for choosing just odd order approaches is given in 
\cite[Chapter~10, \S~3.2]{Davison}.  There it is shown, that boundary 
conditions for even order approximations are much less accurate than 
for odd order models.

Analogous to \defnref{cha_moment:def:moment_operators}, we define
\begin{defs}
    \label{cha_moment:def:restricted_moment_operators}
    The
    ``restricted Intensity to Moment'' operator is
    \begin{equation}
        \begin{split}
            \oM_N \; : \; & \Idomt \to \Mdomt^N \\
                        & \I(t,x,\Omega) \mapsto \vIN(t,x) \;.
        \end{split}
    \end{equation}
    The inverse transformation is given by the ``restricted Moment to
    Intensity'' operator
    \begin{equation}
        \begin{split}
            \oE_N \; : \; & \Mdomt^N \to \IdomtN \subset \Idomt \\
                             & \vIN(t,x) \mapsto \sum_{l=0}^N \sum_{k=-l}^l \vI^k_l(t,x)\sph^k_l(\Omega)=\I_N(t,x,\Omega)\;.
        \end{split}
    \end{equation}
    with
    \begin{equation}
        \IdomtN = \left\{\I_N \in \Idomt \; : \; \I_N(t,x,\Omega) = \sum_{l=0}^N \sum_{k=-l}^l \vI^k_l(t,x) \sph^k_l(\Omega)\right\} \;.
    \end{equation}
\end{defs}

The only difference between the operators $\oE$ and $\oE_N$ is the
restriction on the domain and the range. The restriction of
$\range(\oE_N)$ on $\IdomtN$ ensures  that the injectivity is
inherited from $\oE$. Therefore, $\oE_N$ is still bijective.
$\IdomtN$ is the subspace of $\Idomt$ that contains only those
intensity functions, which can be represented by moments up to order
$N$. Due to the bijectivity of $\oE_N$, working with either the set
of moments up to order $N$ or the approximated intensity
distribution $\I_N(t,x,\Omega)$ is equivalent.

\begin{lem}
    The combined operator
    \begin{equation}
        \begin{split}
            \oP_N \; : \; & \Idomt \to \Idomt \\
                             & \I(t,x, \Omega) \mapsto \oE_N \oM_N \I(t,x,\Omega)
        \end{split}
    \end{equation}
    is a projection.
\end{lem}
\begin{proof}
    We have to show that $\oP_N^2 \I(\Omega) = \oP_N \I(\Omega) =
    \I_N(\Omega)$ holds. Writing down the intensity function as
    series expansion, applying the operators as defined in
    \defnref{cha_moment:def:restricted_moment_operators} and
    additionally using the orthonormality property of the spherical harmonics
    gives the result.
\end{proof}

By using $\oP_N$, we can define the projection onto the orthogonal
complement $\orthoIdomtN \perp \IdomtN$ (with $\Idomt = \IdomtN
\oplus \orthoIdomtN$) by $\tilde \oP_N = \id - \oP_N$. This gives
rise to

\begin{defs}
    The radiative intensity can be decomposed into a component that
    can be described by finitely many moments $\I_N(\Omega)$ and a deviation
    \begin{equation}
        \devIN(\Omega) = \tilde \oP_N \I(\Omega) \;.
    \end{equation}
    This splitting is called deviation decomposition.
\end{defs}

  We call $\devIN(\Omega)$ the deviation because it is the difference
between the full radiative intensity $\I(\Omega)$ and the
component $\I_N(\Omega)$ that can be represented by finitely many
moments:
\begin{equation}
    \begin{split}
        \I(\Omega) & = \oP_N \I(\Omega) + (\id - \oP_N)\I(\Omega) \\
                   & = \I_N(\Omega) + \tilde \oP_N \I(\Omega) \\
                   & =\I_N(\Omega) + \devIN(\Omega) \;.
    \end{split}
\end{equation}

\begin{lem}
    \label{cha_moment:lem:general_decoupled_operator_RTE}
    The RTE can be decomposed into an equivalent coupled system 
    of finitely many moment equations and one deviation equation
    \begin{subequations}
        \label{cha_moment:eq:general_decoupled_operator_RTE}
        \begin{align}
            \label{cha_moment:eq:general_moment_equation_system}
            \frac{1}{c} \partial_t \vIN + \oM_N \oL \oE_N \vIN + \oM_N \oL \devIN &=\vQ_N \;,\\
            \label{cha_moment:eq:general_deviation_equation}
            \frac{1}{c} \partial_t \devIN(\Omega) + \tilde \oP_N \oL \oE_N \vIN + \tilde \oP_N \oL \devIN &=\tilde \oQ_N \;.
        \end{align}
    \end{subequations}
\end{lem}
\begin{proof}
    We start by decomposing the intensity distribution and source terms
    into one component that depends on a finite set of moments and
    second part that is the deviation
    \begin{align}
        \I  &= \ \, \I_N  + \devIN \ = \oE_N \vIN + \devIN  \;.\\
        \oQ  &= \oQ_N  + \tilde \oQ_N = \oE_N \vQ_N + \tilde \oQ_N  \;.
    \end{align}
    By $\vQ_N$ we denote the set of moments generated from the
    source term and $\tilde \oQ_N$ represents the deviation
    part
    \begin{equation}
        \vQ_N=(\mnt^0_0\oQ ,\mnt^{-1}_1\oQ ,\ldots,\mnt^N_N\oQ ,)^T
        \quad \text{and} \quad
        \tilde \oQ_N  = \tilde \oP_N \oQ  \;.
    \end{equation}
    Due to the linearity of the operators, this leads to
    \begin{equation}
        \label{cha_moment:eq:RTE_with_decomposed_intensity}
        \frac{1}{c} \partial_t \left( \oE_N \vIN + \devIN  \right)
          + \oL \oE_N \vIN + \oL \devIN  = \oE_N \vQ_N + \tilde \oQ_N  \;.
    \end{equation}
    Applying the operator $\oM_N$ to this equation gives
    \eqref{cha_moment:eq:general_moment_equation_system}. 
    We have the orthogonality relations $\I_N \bot
    \tilde \I_N$ and $\oQ_N \bot \tilde \oQ_N$ that lead to
    \begin{equation}
        \oM_N(\devIN ) = 0, \qquad \oM_N(\tilde \oQ_N ) = 0 \;.
    \end{equation}
    Using the operator $\tilde \oP_N$ with
    \eqref{cha_moment:eq:RTE_with_decomposed_intensity} leads to
    \eqref{cha_moment:eq:general_deviation_equation}. This is true
    since
    \begin{equation}
        \tilde \oP_N(\oE_N \vIN) = 0, \qquad \tilde \oP_N(\oE_N \vQ_N) = 0 \;.
    \end{equation}
\end{proof}

   We transformed the problem of solving the RTE from solving one
equation in six dimensions to a different problem with a system of
equations for the moments in four dimensions and additionally one
coupled deviation equation that still is six dimensional.  So far we
have not gained anything.  But if we could find a good approximation 
to the deviation, system 
\eqref{cha_moment:eq:general_decoupled_operator_RTE} would simplify
to just the moment equations where the dependence on the deviation
could be treated explicitly.

\section{Closure Approximations}
\label{sec:devapp}

   The classical $P_N$ approximation is obtained by setting the 
deviation to zero
\begin{equation}
    \devIN = 0 \;.
\end{equation}
This is equivalent to assuming that the radiative intensity 
distribution can be written as a finite sum of spherical harmonics.
It is the simplest closure approximation one can make.  Of course, 
in reality this is usually not true and thus, this assumption defines 
the limits of the $P_N$ approach.  The $P_N$ equations in operator 
notation are
\begin{equation}
        \label{cha_moment:eq:operator_pn_equation}
        \frac{1}{c} \partial_t \vIN + \oM_N \oL \oE_N \vIN =\vQ_N.
 \end{equation}
In this work we derive a better approximation of $\devIN$ from the 
deviation equation.  Starting from 
\eqref{cha_moment:eq:general_deviation_equation}, we first assume 
that we can drop the time derivative of the deviation, whereby 
\begin{equation}
    \tilde \oP_N \oL \oE_N \vIN + \tilde \oP_N \oL \devIN =
    \tilde \oQ_N \;.
\end{equation}
Then using the definition of the operator $\oL = \oA + \oK$ and 
assuming the invariance of $\oK$ under the projection $\tilde \oP_N$ 
we obtain
\begin{equation}
    \left(\tilde \oP_N \oA + \oK \right) \devIN 
    = \tilde \oQ_N - \tilde \oP_N \oL \oE_N \vIN \;.
\end{equation}
The invariance assumption is justified if the scattering kernel can
be expanded in terms of spherical harmonics. For the second component on
the right-hand side then holds
\begin{equation}
    \begin{split}
    \tilde \oP_N \oL \oE_N \vIN 
    &= \tilde \oP_N \oA \oE_N \vIN + \oK \tilde \oP_N \oE_N \vIN\\
    & = \tilde \oP_N \oA \oE_N \vIN + \oK \left(\id - \oP_N\right) \oE_N \vIN\\
    & = \tilde \oP_N \oA \oE_N \vIN + \oK \left(\I_N - \I_N\right) 
      = \tilde \oP_N \oA \oE_N\vIN \;.
    \end{split}
\end{equation}
We thereby obtain
\begin{equation}
    \label{cha_moment:ed:formal_equation_for_deviation}
    \devIN = \left(\tilde \oP_N \oA + \oK \right)^{-1}
                     \left(\tilde \oQ_N - \tilde \oP_N \oA \oE_N \vIN\right)
\end{equation}
which is a formal expression for the deviation.  Of course, computing the 
inverse operator $\left(\tilde \oP_N \oA + \oK \right)^{-1}$ is still not 
easier than solving the original transport equation.

   Starting from \eqref{cha_moment:ed:formal_equation_for_deviation}
and using a reformulation gives
\begin{equation}
    \label{cha_moment:ed:formal_equation_for_deviation_transformed_for_neumann}
    \devIN(\Omega) = \left(\id - \left( -\oK^{-1} \tilde \oP_N \oA \right) \right)^{-1}\oK^{-1} 
                     \left(\tilde \oQ_N - \tilde \oP_N \oA \oE_N \vIN\right) \;.
\end{equation}
Recall that we are interested in methods for the transition regime.
Then the collisional physics are more important than the free
transport of the photons.  Hence, we can assume that the $\tilde
\oP_N \oA$ component in
\eqref{cha_moment:ed:formal_equation_for_deviation} which represents
free transport is significantly smaller than the $\oK$ component
which describes absorption and scattering.  We therefore formally use 
Neumann's series to obtain
\begin{equation}
    \label{cha_moment:ed:formal_neumann_series}
    \devIN(\Omega) = \sum_{j=0}^\infty \left( -\oK^{-1} \tilde \oP_N \oA \right)^{j}
    \oK^{-1} \left( \tilde \oQ_N(\Omega) - \tilde \oP_N \oA \oE_N \vIN \right)\;.
\end{equation}
Of course, the operator $\oA$ is not bounded and in general this
series will not converge.  Nevertheless, truncating the expansion
after terms of some order gives an approximation to the deviation
that can be used in the system of moment equations
\eqref{cha_moment:eq:general_moment_equation_system} to improve the
results compared to the $P_N$ method. In this work we will deal with
the approximation that is obtained by taking only the first term of
\eqref{cha_moment:ed:formal_neumann_series}. This leads to
\begin{equation}
    \label{cha_moment:eq:first_order_deviation_equation}
    \devIN = \oK^{-1} \left( \tilde \oQ_N - \tilde \oP_N \oA \oE_N \vIN \right) \;.
\end{equation}
Additionally, we remark that due to the orthogonality of the two
projections $\oP_N \I$ and $\tilde \oP_N\I$ we
have
\begin{equation}
    \oM_N \oK \left( \oK^{-1} \left( \tilde \oQ_N - \tilde \oP_N \oA \oE_N \vIN \right)\right)=0\;.
\end{equation}
Using the deviation approximation
\eqref{cha_moment:eq:first_order_deviation_equation} in
\eqref{cha_moment:eq:general_moment_equation_system} finally leads to the
$D_N$ equations in operator notation
\begin{equation}
        \label{cha_moment:eq:modified_pn_operator_equation}
        \frac{1}{c} \partial_t \vIN + \oM_N \oL \oE_N \vIN
          + \oM_N \oA \left( \oK^{-1} \left( \tilde \oQ_N - \tilde \oP_N \oA \oE_N \vIN \right) \right)
          =\vQ_N.
\end{equation}
Instead of computing the inverse of the operator $\left(\tilde \oP_N
\oA + \oK \right)^{-1}$ we now only have to express the inverse of
the combined absorption and scattering operator $\oK^{-1}$. But
this can be done in a straightforward way, cf.~\ref{apx:B}.

The approximation of the deviation 
can be extended to higher orders. Truncating the series after terms
of order zero gives an additional term with second derivatives in
the moment equations. This is obvious since the operator $\oA$ is
applied twice. Using truncations after terms of higher order
leads to deviations of higher orders and therefore, makes the
equations much more complicated.

\section{An Example}
\label{cha_moment:sec:simplified_example_problem_1d}
In the next section, we are going to develop explicit expressions for the 
introduced operators for radiative transfer in $3D$.  But before we do this, 
to get an understanding how all these operators act on the equations, we take
a closer look on a rather simple problem.  We assume a one-dimensional slab 
geometry, i.e.\  the analyzed radiation field is homogeneous in two 
directions $x_1$ and $x_2$ and also rotationally invariant with respect to 
the axis of propagation.  Then the angular dependence can be expressed in one 
variable $\mu\in[-1,1]$.  For moments of the radiative intensity it holds
\begin{equation}
  \begin{split}
    \int_{\sphere} \overline{\sph^m_n(\Omega)} & \I(\Omega) d\Omega
      = \int_0^{2\pi} \int_{-1}^1 \overline{\sph^m_n(\varphi,\mu)} 
                                  \I(\varphi,\mu) d\mu d\varphi 
\\
    & = (-1)^m \sqrt{\frac{2n+1}{4\pi} \frac{(n-m)!}{(n+m)!}} 
        \left( \int_0^{2\pi} e^{-im\varphi} d\varphi \right) 
        \left( \int_{-1}^1 P^{m}_{n}(\mu) \I(\mu) d\mu \right) \;.
  \end{split}
\end{equation}
But the first integral is zero for every $m\neq 0$ and the set of
relevant moments simplifies to
\begin{equation}
  \label{cha_moment:eq:relevant_moments_1d}
  \vI^m_n = \int_{\sphere} \overline{\sph^m_n(\Omega)} \I(\Omega) d\Omega 
  = \begin{cases}
        \vI^0_n & \text{for} \; m = 0\;, \\
        0 & \text{otherwise} \;.
    \end{cases}
\end{equation}
Due to the homogeneous setup, all derivatives in direction of $x_1$
and $x_2$ vanish and the radiative transfer equation becomes
\begin{equation}
    \label{cha_phy:eq:radiative_transfer_equation_1d}
    \frac{1}{c} \partial_t \I(t,x,\mu) + \mu \partial_{x_3} \I(t,x,\mu)
    + (\scc + \abc) \I(t,x,\mu) - (\oS \I)(t,x,\mu) 
    = \oQ(t,x,\mu) \quad x \in \dom \subset \R \;.
\end{equation}
Moment equations can be generated by applying the operator
$\mnt^0_l$ to this equation:
\begin{equation}
    \frac{1}{c} \partial_t \vI^0_l 
    + \partial_{x_3}
      \left( h_3(0,l) \vI^0_{l+1} + l_3(0,l)\vI^0_{l-1} \right)
    + \tilde \scc_l \vI^0_l = \vQ^0_l \;,
\end{equation}
where
\begin{equation}
  l_3(0,l) = \frac{l}{\sqrt{4l^2-1}} \,, \quad 
  h_3(0,l) = \frac{l+1}{\sqrt{(2l+1)(2l+3)}} \,, \quad 
  \text{and} \quad 
  \tilde \scc_l = \scc + \abc - \frac{\scc}{2}\scc_l \,,
\end{equation}
with $\sigma_l$ being the moments of the scattering kernel, 
cf.\ Appendix \ref{apx:B}.
If in addition we assume an isotropic source, all moments of the
source term of order unequal to zero vanish ($\vQ_l = 0$ for $l>0$).
It can be easily checked, that for the $P_N$ approach this leads to
the following set of equations
\begin{subequations}
    \label{cha_moment:eq:simplified_1d_pn_system}
    \begin{align}
        \frac{1}{c}\partial_t \vI^0_0 + h_3(0,0) \partial_{x_3} \vI^0_{1}
        + \tilde \scc_0 \vI^0_0 = \vQ^0_0 \;, & \\
        \frac{1}{c}\partial_t \vI^0_l + \partial_{x_3}
        h_3(0,l) \partial_{x_3} \vI^0_{l+1} + l_3(0,l) \partial_{x_3} \vI^0_{l-1}
        + \tilde \scc_l \vI^0_l = 0 \;, & \quad l\in\{2,\ldots,N-1 \}\\
        \frac{1}{c}\partial_t \vI^0_N + l_3(0,N) \partial_{x_3} \vI^0_{N-1}
        + \tilde \scc_N \vI^0_N = 0 \;.
    \end{align}
\end{subequations}
For the $D_N$ approach, we have to evaluate the expression
\begin{equation}
    \oM_N \oA \left( \oK^{-1} \left( \tilde \oQ_N - \tilde \oP_N \oA \oE_N \vIN \right) \right) \;.
\end{equation}
Assuming an isotropic source $\oQ$ gives $\tilde \oQ_N(\Omega)=0$.
The moment to intensity operator becomes
\begin{equation}
    \oE_N \vIN = \sum_{l=0}^N \sph^0_l \vI^0_l \;,
\end{equation}
and by using the projection property of $\tilde \oP_N$ we get
\begin{equation}
    \begin{split}
        \tilde \oP_N \oA \oE_N \vIN
          &= \tilde \oP_N \sum_{l=0}^N \Omega_3 \sph^0_l(\Omega) \partial_{x_3} \vI^0_l \\
          &= \tilde \oP_N \sum_{l=0}^N \left( h_3(0,l) \sph^0_{l+1}(\Omega) + l_3(0,l) \sph^0_{l-1}(\Omega) \right) \partial_{x_3} \vI^0_l \\
          &= h_3(0,N) \sph^0_{N+1}(\Omega) \partial_{x_3} \vI^0_N \;.
    \end{split}
\end{equation}
Applying $\oA$ and $\oK^{-1}$ to this expression yields
\begin{equation}
    \begin{split}
        \mnt^0_n \oA \oK^{-1} \tilde \oP_N \oA \oE_N \vIN
          &= \mnt^0_n \left( \partial_{x_3} \left(\frac{1}{\tilde \scc_{N+1}}
             h_3(0,N) \Omega_3 \sph^0_{N+1}(\Omega) \partial_{x_3} \vI^0_N \right) \right) \\
          &= \mnt^0_n \biggl( \partial_{x_3} \biggl(\frac{1}{\tilde \scc_{N+1}}
             h_3(0,N) \bigl( h_3(0,N+1)\sph^0_{N+2}(\Omega) \\
          & \qquad \qquad \qquad \qquad \qquad \qquad + l_3(0,N+1)\sph^0_{N}(\Omega) \bigr) \partial_{x_3} \vI^0_N \biggr) \biggr) \\
          &= \begin{cases}
                \partial_{x_3} \left(\frac{1}{\tilde \scc_{N+1}} h_3(0,N) l_3(0,N+1) \partial_{x_3} \vI^0_N \right) & \text{for} \quad n=N \;, \\
                0 & \text{otherwise} \;.
             \end{cases}
    \end{split}
\end{equation}
From thes calculations we see that the $D_N$ equations
differ from the $P_N$ equations
\eqref{cha_moment:eq:simplified_1d_pn_system} only in the the
equations for the moment of order $N$.  The $D_N$ system
finally reads
\begin{subequations}
    \label{cha_moment:eq:simplified_1d_modified_pn_system}
    \begin{align}
        \frac{1}{c}\partial_t \vI^0_0 + h_3(0,0) \partial_{x_3} \vI^0_{1}
        + \tilde \scc_0 \vI^0_0 = \vQ^0_0 \;, & \\
        \frac{1}{c}\partial_t \vI^0_l + \partial_{x_3}
        h_3(0,l) \partial_{x_3} \vI^0_{l+1} + l_3(0,l) \partial_{x_3} \vI^0_{l-1}
        + \tilde \scc_l \vI^0_l = 0 \;, & \\
        \frac{1}{c}\partial_t \vI^0_N + l_3(0,N) \partial_{x_3} \vI^0_{N-1}
        - \partial_{x_3} \left(\frac{1}{\tilde \scc_{N+1}} h_3(0,N) l_3(0,N+1) \partial_{x_3} \vI^0_N \right)
        + \tilde \scc_N \vI^0_N = 0 \;.
    \end{align}
\end{subequations}
The correction term of the $D_N$ equation is of diffusive nature and 
thus adds a stabilizing component to the $P_N$ equations.

\begin{rem}
    \label{cha_moment:rem:interpret_difference_pn_modified_pn_1d}
    The additional term also can be interpreted in a different way. If
    we take the moment equation of order $N+1$, neglect moments of order
    $N+2$ and the time derivative and solve this equation for the moment
    $\vI^0_{N+1}$ we get
    \begin{equation}
        \label{cha_moment:eq:simplified_interpretation_modified_pn}
        \vI^0_{N+1} = -\frac{1}{\tilde \scc_{N+1}} l_3(0,N+1) \partial_{x_3} \vI^0_{N} \;.
    \end{equation}
    Inserting this term as approximation for $\vI^0_{N+1}$ into the
    equation for the moment $\vI^0_N$ gives exactly the equation for the
    moment of order $N$ in the $D_N$ equations.  Therefore, at least in 1D 
    there is a simple way how the new model equations can be derived.
\end{rem}
    
\section{Explicit Operators for $P_N$}
\label{cha_moment:sec:analysis_operators_pn_approximation}
In the previous sections, we developed an operator approach to solve 
the RTE by moment methods.  Now we take a more detailed look at these 
operators and analyze their structure.  Also the results presented 
here are relevant to develop numerical methods for solving the RTE 
with the help of $P_N$ and $D_N$ equations.  Most of the notation 
used here is introduced in the Appendix.

In this section we will assume that the source term $\oQ$ is
isotropic and thus does not depend on the direction of the radiation
$\Omega$. Then, due to the orthogonality of the spherical harmonics,
all directional moments of $\oQ(t,x)$ of order equal or higher than
one vanish and we get
\begin{equation}
    \oQ(t, x) = \sum_{l=0}^\infty \sum_{k=-l}^l \sph^k_l(\Omega) \mnt^k_l \oQ(t,x)
           = \frac{1}{\sqrt{4\pi}} \vQ^0_0(t,x) \;.
\end{equation}
For the vector of moments of the source and its deviation this leads
to
\begin{equation}
    \label{cha_moment:eq:properties_moments_isotropic_source}
    \vQ_N(t,x) = \left( \sqrt{4\pi}\left( \abc \B(T) + \Q(t,x) \right),0,0,\ldots \right)^T \qquad \text{and} \qquad \tilde
    \oQ_N(t,x,\Omega) = 0 \;.
\end{equation}
Next, we will analyze the expression
\begin{equation}
    \oM_N \oL \oE_N \vI_N = \oM_N \oA \oE_N \vI_N + \oM_N \oK \oE_N \vI_N
\end{equation}
which appears in the $P_N$ and $D_N$ approaches.  The analysis will 
be performed separately for the transport ($\oA$) and the 
scattering/absorption ($\oK$) component. We will start with the
transport term $\oM_N \oA \oE_N \vI_N$.

By analyzing one single moment equation of order $n$ and degree $m$
we get
\begin{equation}
    \begin{split}
        \mnt^m_n \oA \oE_N \vI_N
          &= \mnt^m_n \left(\sum_{r=1}^3 \partial_{x_r} \Omega_r \I_N(\Omega)\right)
          = \mnt^m_n \left(\sum_{r=1}^3 \partial_{x_r} \Omega_r \left( \sum_{l=0}^N \sum_{k=-l}^l \sph^k_l(\Omega) \vI^k_l\right)\right) \\
          &= \sum_{r=1}^3 \partial_{x_r} \left( \sum_{l=0}^N \sum_{k=-l}^l \left(\mnt^m_n \Omega_r \sph^k_l(\Omega)\right) \vI^k_l\right) \\
    \end{split}
\end{equation}
The inner sum can be written as a scalar product
\begin{equation}
    \sum_{l=0}^N \sum_{k=-l}^l \left(\mnt^m_n \Omega_r \sph^k_l(\Omega)\right)\vI^k_l
    = \left\langle \mnt^m_n \Omega_r \vY_N , \vI_N \right\rangle
    = \left( \mnt^m_n \Omega_r \vY_N\right)^T  \vI_N  \;,
\end{equation}
where $\vY_N$ denotes the vector of spherical harmonics as
introduced in \defnref{apx_sph:def:vector_of_spherical_harmonics}.

Expressing the first component of this scalar product with the help
of relation \eqref{apx_sph:eq:omega_times_sph_matrix_notation} and
orthogonality relations for spherical harmonics
gives for one
component of the vector
\begin{equation}
    \begin{split}
        \mnt^m_n \Omega_r \sph^i_j
          &= \mnt^m_n \left( \gamma_r \left( \left(\ve^{N}_{(i,j)}\right)^T \hat L_{x_r}^{N} \vY_N
                + \left(\ve^{N+1}_{(i,j)}\right)^T \hat U_{x_r}^{N+1} \vY_{N+1} \right) \right)\\
          &= \gamma_r \left( \left(\ve^{N}_{(i,j)}\right)^T \hat L_{x_r}^{N} \mnt^m_n \vY_N
                + \left(\ve^{N+1}_{(i,j)}\right)^T \hat U_{x_r}^{N+1} \mnt^m_n \vY_{N+1} \right)\\
          &= \gamma_r \left( \left(\ve^{N}_{(i,j)}\right)^T \hat L_{x_r}^{N} \ve^{N}_{(m,n)}
                + \left(\ve^{N+1}_{(i,j)}\right)^T \hat U_{x_r}^{N+1} \ve^{N+1}_{(m,n)}\right)\;.
    \end{split}
\end{equation}
In our situation holds $j,n\in\{0,\ldots,N\}$ and therefore the unit
vectors $\ve^{N+1}_{(i,j)}$ and $\ve^{N+1}_{(m,n)}$ are always zero
in the last $2N+3$ components. Hence we can reduce the dimension of
the last term without losing any information and end up with
\begin{equation}
    \mnt^m_n \Omega_r \vY_N = \gamma_r \left( \hat L_{x_r}^{N} + \hat U_{x_r}^{N} \right) \ve^{N}_{(m,n)} \;.
\end{equation}
Finally we get for one moment equation
\begin{equation}
    \begin{split}
        \mnt^m_n \oA \oE_N \vI_N
          &= \sum_{r=1}^3 \partial_{x_r} \left\langle \mnt^m_n \Omega_r \vY_N , \vI_N \right\rangle \\
          &= \sum_{r=1}^3 \left( \left(\gamma_r \left( \hat L_{x_r}^{N} + \hat U_{x_r}^{N} \right) \ve^{N}_{(m,n)} \right)^T \partial_{x_r} \vI_N \right) \\
          &= \left(\ve^{N}_{(m,n)}\right)^T \sum_{r=1}^3 \gamma_r \left( \left( \hat L_{x_r}^{N}\right)^T + \left(\hat U_{x_r}^{N} \right)^T  \right) \partial_{x_r} \vI_N\;.
    \end{split}
\end{equation}
Taking advantage of the symmetry properties provided in
\lemref{apx_sph:lem:symmetrie_properties_submatrices_sph} gives rise
to the definitions
\begin{equation}
    \label{cha_moment:eq:structure_complex_system_matrices}
    \oC_{x_1} = \hat L_{x_1}^{N} + \hat U_{x_1}^{N}\;, \quad
    \oC_{x_2} = -\hat L_{x_2}^{N} - \hat U_{x_2}^{N}\;, \quad
    \oC_{x_3} = \hat L_{x_3}^{N} + \hat U_{x_3}^{N} \;.
\end{equation}
For the set of all moment equations holds
\begin{equation}
    \label{cha_moment:eq:moment_term_transport_component}
    \oM_N \oA \oE_N \vI_N
      = \sum_{r=1}^3 \gamma_r \oC_{x_r} \partial_{x_r}\vI_N
      = \frac{1}{2} \oC_{x_1} \partial_{x_1}\vI_N + \frac{i}{2} \oC_{x_2}
             \partial_{x_2}\vI_N + \oC_{x_3} \partial_{x_3}\vI_N \;.
\end{equation}
Let us now analyze the scattering and absorption component $\oM_N
\oK \oE_N \vI_N$. As defined in
\eqref{cha_model:eq:operator_definitions_RTE} it decomposes into
\begin{equation}
    \oK \oE_N \vI_N = \left(\abc + \scc \right) \oE_N \vI_N - \oS \oE_N
    \vI_N
\end{equation}
and from \apxref{apx:B} (especially from
\eqref{apx_scat:eq:single_moment_absorb_and_scattering}) we get for
one single moment and $l\leq N$
\begin{equation}
    \label{cha_moment:eq:single_moment_scattering_absorption_opp}
    \mnt^k_l \oK \oE_N \vI_N = \tilde \scc_l \vI^k_l.
\end{equation}
For the set of all moment equations we end up with
\begin{equation}
    \label{cha_moment:eq:moment_scattering_absorption_opp}
    \oM_N \oK \oE_N \vI_N = \Sigma_N \vI_N
\end{equation}
where we define
\begin{equation}
    \Sigma =
    \begin{pmatrix}
      \tilde \Sigma_0 &  &  \\
       & \ddots &  \\
       &  & \tilde \Sigma_N \\
    \end{pmatrix}
\end{equation}
with diagonal submatrices
\begin{equation}
    \R^{2j+1}\times\R^{2j+1} \ni \tilde \Sigma_j = \tilde \scc_j \id
    \;.
\end{equation}
Finally, the full set of moment equations for a $P_N$ approximation
reads
    \begin{equation}
        \label{cha_moment:eq:complex_valued_pn_equations}
        \frac{1}{c}\partial_t \vI_N
            + \frac{1}{2} \oC_{x_1} \partial_{x_1} \vI_N
            + \frac{i}{2} \oC_{x_2} \partial_{x_2} \vI_N
            + \oC_{x_3} \partial_{x_3} \vI_N
            + \Sigma_N \vI_N = \vQ_N
    \end{equation}
with the vector of moments of the source term $\vQ_N$. This
representation is equivalent to the operator notation given in
\eqref{cha_moment:eq:operator_pn_equation} but the operators are
made explicit by their matrix representation.
\begin{prop}
    \label{cha_moment:prop:symmetric_complex_system_matrices}
    The system matrices $\oC_{x_1}$ and $\oC_{x_3}$ are
    symmetric. $\oC_{x_2}$ is skew symmetric. This means that the $P_N$ equations are hyperbolic.
\end{prop}
\begin{proof}
    Due to the symmetry results from
    \lemref{apx_sph:lem:symmetrie_properties_submatrices_sph} and the
    structure of the system matrices given in
    \eqref{cha_moment:eq:structure_complex_system_matrices} the
    result is obvious.
\end{proof}

\section{Explicit Diffusion Operators for $D_N$}
\label{cha_moment:sec:analysis_operators_modified_pn_approximation}
The $D_N$ equations \eqref{cha_moment:eq:modified_pn_operator_equation} 
have one additional component that was not already treated for the $P_N$
equations in \secref{cha_moment:sec:analysis_operators_pn_approximation}:
\begin{equation}
     \oM_N \oA \left( \oK^{-1} \left( \tilde \oQ_N - \tilde \oP_N \oA \oE_N \vIN \right) \right) \;.
\end{equation}
In this section we investigate the properties of this term and
develop a representation that fits into the matrix framework
developed for the $P_N$ equations in
\eqref{cha_moment:eq:complex_valued_pn_equations}.

Dealing only with isotropic source terms results in vanishing
deviations $\tilde \oQ_N$ of the source term (see
\eqref{cha_moment:eq:properties_moments_isotropic_source}) and
simplifies the problem to analyzing
\begin{equation}
    \label{cha_moment:eq:modified_pn_additional_term_operator_notation}
      -\oM_N \oA \oK^{-1} \tilde \oP_N \oA \oE_N \vIN\;.
\end{equation}
We start with
\begin{equation}
    \begin{split}
        \tilde \oP_N \oA \oE_N \vIN
          &= \tilde \oP_N \oA \sum_{l=0}^N \sum_{k=-l}^l \sph^k_l(\Omega) \vI^k_l
           = \tilde \oP_N \oA \left\langle \vY_N, \vI_N \right\rangle\\
          &= \tilde \oP_N \sum_{s=1}^3 \left\langle \Omega_s \vY_N, \partial_{x_s} \vI_N \right\rangle \;.
    \end{split}
\end{equation}
Formally the projection $\tilde \oP_N$ is only defined as an
operator that acts on functions from $\Idomt$ into $\Idomt$. When
working with vectors of functions we just apply the projection to
every component.

Using the decomposition of the orthogonal projection into $\tilde
\oP_N = \id - \oP_N$ and taking a component wise look at $\oP_N
\Omega_s \vY_N$ leads us to
\begin{equation}
    \begin{split}
        \oP_N \Omega_s \sph^i_j
          &= \oP_N \left( \gamma_s \left( \left(\ve^{N}_{(i,j)}\right)^T \hat L_{x_s}^{N} \vY_N + \left(\ve^{N+1}_{(i,j)}\right)^T \hat U_{x_s}^{N+1} \vY_{N+1}\right)  \right)\\
          &= \gamma_s \left( \left(\ve^{N}_{(i,j)}\right)^T \hat L_{x_s}^{N} \oP_N \vY_N + \left(\ve^{N+1}_{(i,j)}\right)^T \hat U_{x_s}^{N+1} \oP_N \vY_{N+1}\right) \;.
    \end{split}
\end{equation}
Applying the projection operator $\oP_N$ on spherical harmonics up
to order $N$ is an identity operation and it holds $\oP_N \sph^m_n =
\sph^m_n$ for $n\leq N$. If we apply $\oP_N$ on spherical harmonic
of order larger than $N$ the result is $\oP_N \sph^m_n = 0$ for
$n>N$.
\begin{equation}
    \oP_N \vY_{N+1} = (\sph^0_0, \sph^{-1}_1,\ldots,\sph^{N-1}_N,
    \sph^N_N, \underbrace{0, \ldots, 0}_{2N+3 \; \text{zeros}})^T =\vY_{N+1}^\text{cut}\;.
\end{equation}
Using these results for the orthogonal projection leads to
\begin{equation}
    \begin{split}
        \tilde \oP_N \Omega_s \sph^i_j
          &= (\id - \oP_N) \Omega_s \vY^i_j\\
          &= \gamma_s \left( \left(\ve^{N}_{(i,j)}\right)^T \hat L_{x_s}^{N} \left(\vY_N - \vY_N\right)
             + \left(\ve^{N+1}_{(i,j)}\right)^T \hat U_{x_s}^{N+1} \left( \vY_{N+1} - \vY_{N+1}^\text{cut} \right) \right) \\
          &= \gamma_s \left(\ve^{N+1}_{(i,j)}\right)^T \hat U_{x_s}^{N+1} \left( 0,\ldots,0, \sph^{-N-1}_{N+1},\ldots,\sph^{N+1}_{N+1}\right)^T \;.
    \end{split}
\end{equation}
By using the fact that $j\leq N$ and the special structure of the
matrices $U_{x_s}^{N+1}$ (see
\eqref{apx_sph:eq:structure_complex_system_matrices_subparts}) we
finally obtain for the full vector of spherical harmonics $\vY_N$
\begin{equation}
    \tilde \oP_N \Omega_s \vY_N = \gamma_s Z^N_\text{cut} \hat U_{x_s}^{N+1} \vY_{N+1}
\end{equation}
with
\begin{equation}
    Z^N_\text{cut} = \left(
    \begin{array}{ccc|c|c}
      0 &        &   & 0      & 0      \\
        & \ddots &   & \vdots & \vdots \\
        &        & 0 & 0      & 0      \\
      \hline
      0 & \cdots & 0 & \id    & 0      \\
    \end{array} \right)
\end{equation}
It holds $Z^N_\text{cut} \in \R^{(N+1)^2} \times \R^{(N+2)^2}$ and
the block with the identity matrix is located in the rows $(N^2+1)$
till $(N+1)^2$ and the columns $(N^2+1)$ till $(N+1)^2$. Finally we
get
\begin{equation}
    \begin{split}
        \tilde \oP_N \oA \oE_N \vIN
          &= \sum_{s=1}^3 \gamma_s \left(Z^N_\text{cut} \hat U_{x_s}^{N+1} \vY_{N+1}\right)^T \partial_{x_s} \vI_N \\
          &= \left(\vY_{N+1}\right)^T \sum_{s=1}^3 \gamma_s \left(\hat U_{x_s}^{N+1}\right)^T \left(Z^N_\text{cut}\right)^T \partial_{x_s} \vI_N \;.
    \end{split}
\end{equation}
Here we see, that only the moments of order $N$ influence the
additional term in the $D_N$ equations.

Applying the inverse of the combined
scattering and absorption operator $\oK^{-1}$ and
using results from \apxref{apx:B} gives
\begin{equation}
    \oK^{-1} \tilde \oP_N \oA \oE_N \vIN =
    \frac{1}{\tilde \scc_{N+1}} \left(\vY_{N+1}\right)^T \sum_{s=1}^3 \gamma_s \left(\hat U_{x_s}^{N+1}\right)^T \left(Z^N_\text{cut}\right)^T \partial_{x_s} \vI_N \;.
\end{equation}
By considering the complete term from
\eqref{cha_moment:eq:modified_pn_additional_term_operator_notation},
for one single moment equation holds
\begin{equation}
    \label{cha_moment:eq:prestep1_modified_pn_derivation}
    \begin{split}
        \mnt^m_n \oA \oK^{-1} & \tilde \oP_N \oA \oE_N \vIN \\
          &= \mnt^m_n \oA \left( \frac{1}{\tilde \scc_{N+1}} \left(\vY_{N+1}\right)^T \sum_{s=1}^3 \gamma_s \left(\hat U_{x_s}^{N+1}\right)^T \left(Z^N_\text{cut}\right)^T \partial_{x_s} \vI_N \right) \\
          &= \sum_{r=1}^3 \partial_{x_r} \left( \frac{1}{\tilde \scc_{N+1}} \left(\mnt^m_n \Omega_r \left(\vY_{N+1}\right)^T \right) \sum_{s=1}^3 \gamma_s \left(\hat U_{x_s}^{N+1}\right)^T \left(Z^N_\text{cut}\right)^T \partial_{x_s} \vI_N \right) \;,
    \end{split}
\end{equation}
where we have to find an expression for $\mnt^m_n \Omega_r
\left(\vY_{N+1}\right)^T$. Again, we start be analyzing on single
component of the vector with $n\in\{0,\ldots,N\}$,
$m\in\{-n,\ldots,n\}$, $j\in\{0,\ldots,N+1\}$ and
$i\in\{-j,\ldots,j\}$
\begin{equation}
    \begin{split}
        \mnt^m_n \Omega_r \sph^i_j
          &= \mnt^m_n \left( \gamma_r \left( \left(\ve^{N+1}_{(i,j)}\right)^T \hat L_{x_r}^{N+1} \vY_{N+1} + \left(\ve^{N+2}_{(i,j)}\right)^T \hat U_{x_r}^{N+2} \vY_{N+2}\right) \right) \\
          &= \gamma_r \left( \left(\ve^{N+1}_{(i,j)}\right)^T \hat L_{x_r}^{N+1} \left(\mnt^m_n \vY_{N+1}\right) + \left(\ve^{N+2}_{(i,j)}\right)^T \hat U_{x_r}^{N+2} \left(\mnt^m_n \vY_{N+2}\right)\right)\\
          &= \gamma_r \left( \left(\ve^{N+1}_{(i,j)}\right)^T \hat L_{x_r}^{N+1} \ve^{N+1}_{(m,n)} + \left(\ve^{N+2}_{(i,j)}\right)^T \hat U_{x_r}^{N+2} \ve^{N+2}_{(m,n)} \right)\;.
    \end{split}
\end{equation}
Due to the orthogonality of the spherical harmonics the expressions
$\mnt^m_n \vY_{N+1}$ and $\mnt^m_n \vY_{N+2}$ lead to unit vectors
$\ve^{N+1}_{(m,n)}$ and $\ve^{N+2}_{(m,n)}$. But since $n \leq N$
the component that is one is always located in the first $(N+1)^2$
components. For the same reason ($j\leq N+1$) the last $2N+4$
components of $\ve^{N+2}_{(i,j)}$ are always zero and we therefore
can neglect the last $2N+4$ rows and columns of $\hat
U_{x_r}^{N+2}$. Due to the structure of these matrices this is
equivalent to writing
\begin{equation}
    \begin{split}
        \mnt^m_n \Omega_r \sph^i_j
          &= \gamma_r \left(\ve^{N+1}_{(i,j)}\right)^T \left( \hat L_{x_r}^{N+1} + \hat U_{x_r}^{N+1} \right) \ve^{N+1}_{(m,n)} \;.
    \end{split}
\end{equation}
For a complete vector of spherical harmonics we obtain
\begin{equation}
    \begin{split}
        \mnt^m_n \Omega_r \vY_{N+1}
          &= \gamma_r \left( \hat L_{x_r}^{N+1} + \hat U_{x_r}^{N+1} \right) \ve^{N+1}_{(m,n)} \;.
    \end{split}
\end{equation}
For the transposed expression that is relevant in
\eqref{cha_moment:eq:prestep1_modified_pn_derivation} we end up with
\begin{equation}
    \begin{split}
        \mnt^m_n \Omega_r \left(\vY_{N+1}\right)^T
          &= \gamma_r \left(\ve^{N+1}_{(m,n)}\right)^T \left( \hat L_{x_r}^{N+1} + \hat U_{x_r}^{N+1} \right)^T \;.
    \end{split}
\end{equation}
If we use the fact $\hat L_{x_r}^k \hat L_{x_s}^k=0$ and $\hat
U_{x_r}^k \hat U_{x_s}^k=0$ for $k\in \N$ we get
\begin{equation}
    \label{cha_moment:eq:additional_term_in_modified_pn}
    \begin{split}
        & \mnt^m_n \oA \oK^{-1} \tilde \oP_N \oA \oE_N \vIN \\
          &= \left(\ve^{N+1}_{(m,n)}\right)^T \sum_{r=1}^3 \gamma_r \partial_{x_r} \left( \frac{1}{\tilde \scc_{N+1}} \left( \hat L_{x_r}^{N+1} \right)^T \sum_{s=1}^3 \gamma_s \left(\hat U_{x_s}^{N+1}\right)^T \left(Z^N_\text{cut}\right)^T \partial_{x_s} \vI_N \right) \;, \\
    \end{split}
\end{equation}
or if we use
\begin{equation}
    \left(\R^{(N+1)^2}\right)^2 \ni \oD_{r,s} = \gamma_r \gamma_s
    Z^N_\text{restrict}
    \left( \hat L_{x_r}^{N+1} \right)^T 
    \left( \hat U_{x_s}^{N+1} \right)^T 
    \left( Z^N_\text{cut} \right)^T \;,
\end{equation}
with
\begin{equation}
    \R^{(N+1)^2} \times \R^{(N+2)^2} \ni Z^N_\text{restrict} =
    \left(
    \begin{array}{ccc|ccc}
        1 &        &   & 0      & \cdots & 0 \\
          & \ddots &   & \vdots &        & \vdots \\
          &        & 1 & 0      & \cdots & 0 \\
    \end{array} \right) \;,
\end{equation}
we end up with
\begin{equation}
    -\oM_N \oA \oK^{-1} \tilde \oP_N \oA \oE_N \vIN
      = -\sum_{r=1}^3 \partial_{x_r} \left( \frac{1}{\tilde \scc_{N+1}} \sum_{s=1}^3 \oD_{r,s} \partial_{x_s} \vI_N \right) \;.
\end{equation}
for the set of all moment equations.

Finally, the moment equations for the $D_N$ approach read
    \begin{equation}
        \label{cha_moment:eq:complex_valued_modified_pn_equations}
        \begin{split}
            \frac{1}{c}\partial_t \vI_N
                + \frac{1}{2} \oC_{x_1} \partial_{x_1} \vI_N
                &+ \frac{i}{2} \oC_{x_2} \partial_{x_2} \vI_N
                + \oC_{x_3} \partial_{x_3} \vI_N \\
                &\qquad \qquad -\sum_{r=1}^3 \partial_{x_r} \left( \frac{1}{\tilde \scc_{N+1}} \sum_{s=1}^3 \oD_{r,s} \partial_{x_s} \vI_N \right)
                + \Sigma_N \vI_N = \vQ_N \;.
        \end{split}
    \end{equation}
Due to the special structure of the matrices $\hat U_{x_s}^{N+1}$
and $\hat L_{x_s}^{N+1}$ from
\eqref{apx_sph:eq:structure_complex_system_matrices_subparts} we get
for the product
\begin{equation}
      \Big( \hat L_{x_r}^{N+1} \Big)^T \! 
      \Big( \hat U_{x_s}^{N+1} \Big)^T \!
      = \Big( \hat U_{x_r}^{N+1}  \hat L_{x_s}^{N+1} \Big)^T \!
      = \begin{pmatrix}
              \big( U_{x_r}^1 L_{x_s}^1 \big)^T  &    &    &    \\
                                           &  \ddots  &    &    \\
            &    &  \big(U_{x_r}^{N+1} L_{x_s}^{N+1} \big)^T  &    \\
            &    &                                            &  0   
        \end{pmatrix} \;,
\end{equation}
For the complete matrix $\oD_{r,s}$ we obtain
\begin{equation}
    \begin{split}
        \oD_{r,s}
          &= \gamma_r \gamma_s Z^N_\text{restrict}
             \begin{pmatrix}
               \left(U_{x_r}^1 L_{x_s}^1\right)^T &        &                             &  \\
                                   & \ddots &                             &  \\
                                   &        & \left(U_{x_r}^{N+1} L_{x_s}^{N+1}\right)^T &  \\
                                   &        &                             & 0\\
             \end{pmatrix}
             \left(Z^N_\text{cut}\right)^T \\
          &= \begin{pmatrix}
               0 &        &   &                                                              \\
                 & \ddots &   &                                                              \\
                 &        & 0 &                                                              \\
                 &        &   & \gamma_r \gamma_s \left(U_{x_r}^{N+1} L_{x_s}^{N+1}\right)^T \\
             \end{pmatrix} \;.
    \end{split}
\end{equation}
From this matrix we see, that the $D_N$ equations differ
from the usual $P_N$ equations only in the equations for the moments
of order $N$.

\section{Numerical Results}
\label{cha_numerics:sec:numerical_results}
In this section we present some numerical simulations of two radiative 
transport problems.  We compare the $D_N$ method with the standard 
$P_N$ method.  We first investigate a one-dimensional test problem that can 
be solved analytically.  The second test problem compares the two methods in 
a strongly inhomogeneous two-dimensional medium.
%
%
%
\subsection{$P_N$ and $D_N$ models vs. benchmark solution}
\label{cha_numerics:ssec:numerical_results_benchmark}
In this numerical test we compare the $P_N$ and $D_N$ 
approximations to an analytic benchmark solution.  The details regarding 
the benchmark solution can be found in \cite{Su1997}.

The physical setting is initially a cold, homogeneous, and infinite 
isotropically scattering medium.  An internal slab radiation source 
is switched on at time $t=0$ and off at $t=t_\text{end}$.   Because 
this problem has slab symmetry, it can be described by the $1D$
radiative transfer and energy transfer equations
\begin{subequations}
    \begin{equation}
        \frac{1}{c} \partial_t \I + \mu \partial_x \I 
        + (\scc+\abc)\I - \oS \I = \abc a T^4 + \oQ \;,
    \end{equation}
    \begin{equation}
        \varrho c_v \frac{\partial T}{\partial t} 
        = \abc \left(\int_{-1}^1 \I(\mu') \, d\mu' - 2 aT^4 \right)  \;.
    \end{equation}
\end{subequations}
Assuming that the coefficients of absorption and scattering are 
constant and that
\begin{equation}
    c_v = T^3 \;,
\end{equation}
the usually nonlinear system becomes linear in the variables $\I$ and 
$T^4$.  For convenience we set $c=1$ and $\varrho=1$.

   We impose the boundary conditions 
\begin{equation}
    \lim_{x\to \pm \infty} \I(t,x,\mu_\text{in}) = 0 \;, \qquad 
    \lim_{x\to \pm \infty} T(t,x)=0 \;,
\end{equation}
and the initial conditions
\begin{equation}
    \I(t=0,x,\mu) = 0 \;, \qquad T(t=0,x) = 0 \;.
\end{equation}
A uniform isotropic radiation source is turned on in the slab 
$[-x_0,x_0]$ over the time interval $t\in[0,t_\text{end}]$.  It can 
be described by
\begin{equation}
    \Q(t,x,\mu) =
    \begin{cases}
        \frac{1}{2x_0} 
        & \text{for} \quad x \in [-x_0,x_0] \quad
          \text{and} \quad t \in [0,t_\text{end}] \;,\\
        0 & \text{otherwise} \;.
    \end{cases}
\end{equation}
Several benchmark solutions were provided in \cite{Su1997} for various 
values of the absorption and scattering coefficients.  For our 
comparisons we will use the one with
\begin{equation}
    \abc = 1 \;, \quad \scc = 0 \;, \quad x_0 = 0.5 \;, \quad
    t_\text{end} = 10 \;.
\end{equation}
The comparison of this benchmark solution with the $P_N$ and $D_N$ 
solutions is given in \figref{cha_numerics:fig:suolson_results}.
The results have been obtained with a kinetic scheme for the transport part of 
the equation and a standard finite differences discretization of the difusion terms. 
The grid has been refined until numerical convergence was observed.
\begin{figure}[htbp]
    \centering
    \subbottom[$P_N$ approximations.]
        {\label{cha_numerics:fig:pn_solution_suolson}
        \includegraphicschoice{width=0.425\textwidth}{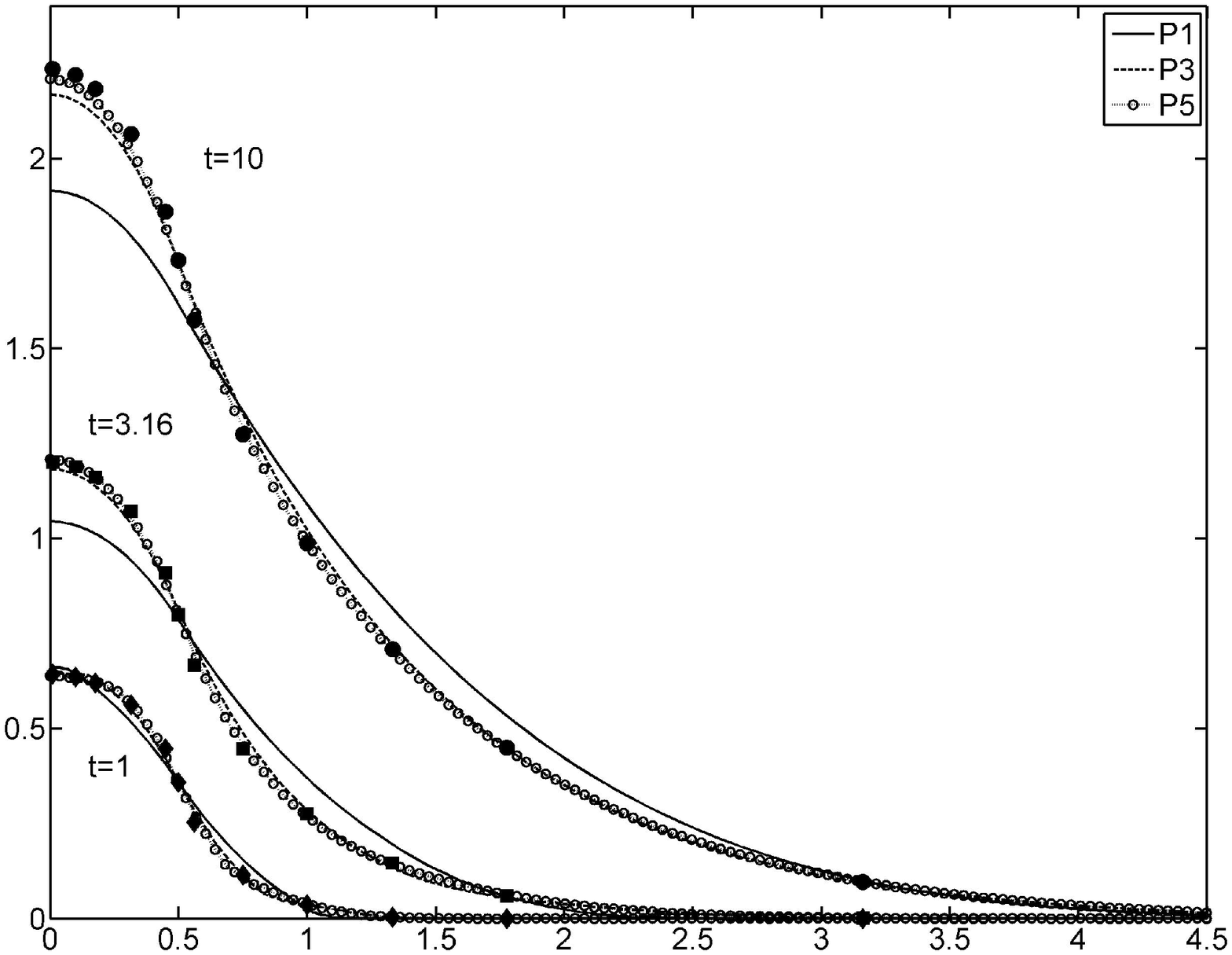}}\hspace{1cm}
    \subbottom[$D_N$ approximations.]
        {\label{cha_numerics:fig:mpn_solution_suolson}
        \includegraphicschoice{width=0.425\textwidth}{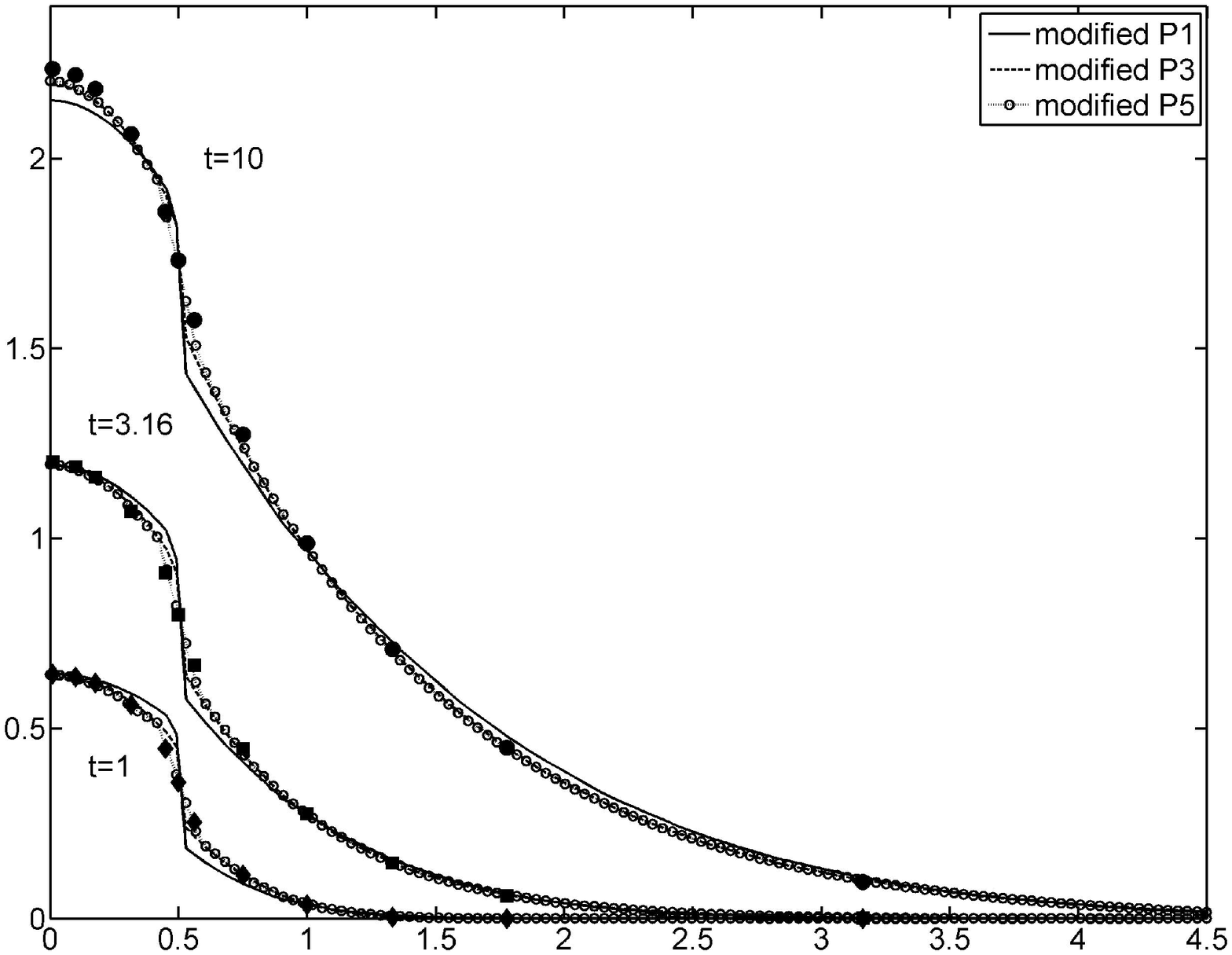}}
    \subbottom[Comparison: $D_1$ vs. $P_1$ approximations.]
        {\label{cha_numerics:fig:compare_mp1_p1_solution_suolson}
        \includegraphicschoice{width=0.425\textwidth}{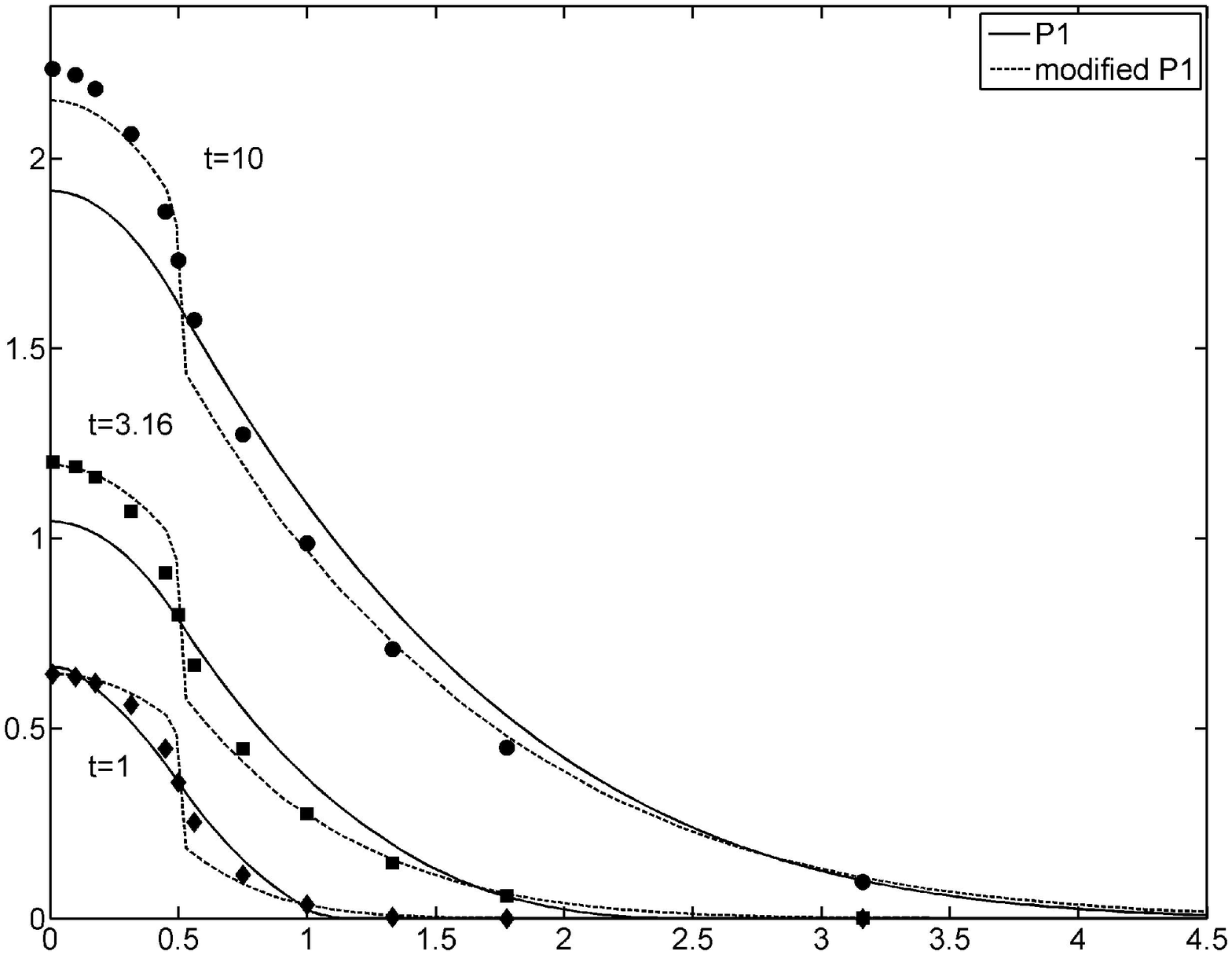}}\hspace{1cm}
    \subbottom[Comparison: $D_3$ vs. $P_3$ approximations.]
        {\label{cha_numerics:fig:compare_mp3_p3_solution_suolson}
        \includegraphicschoice{width=0.425\textwidth}{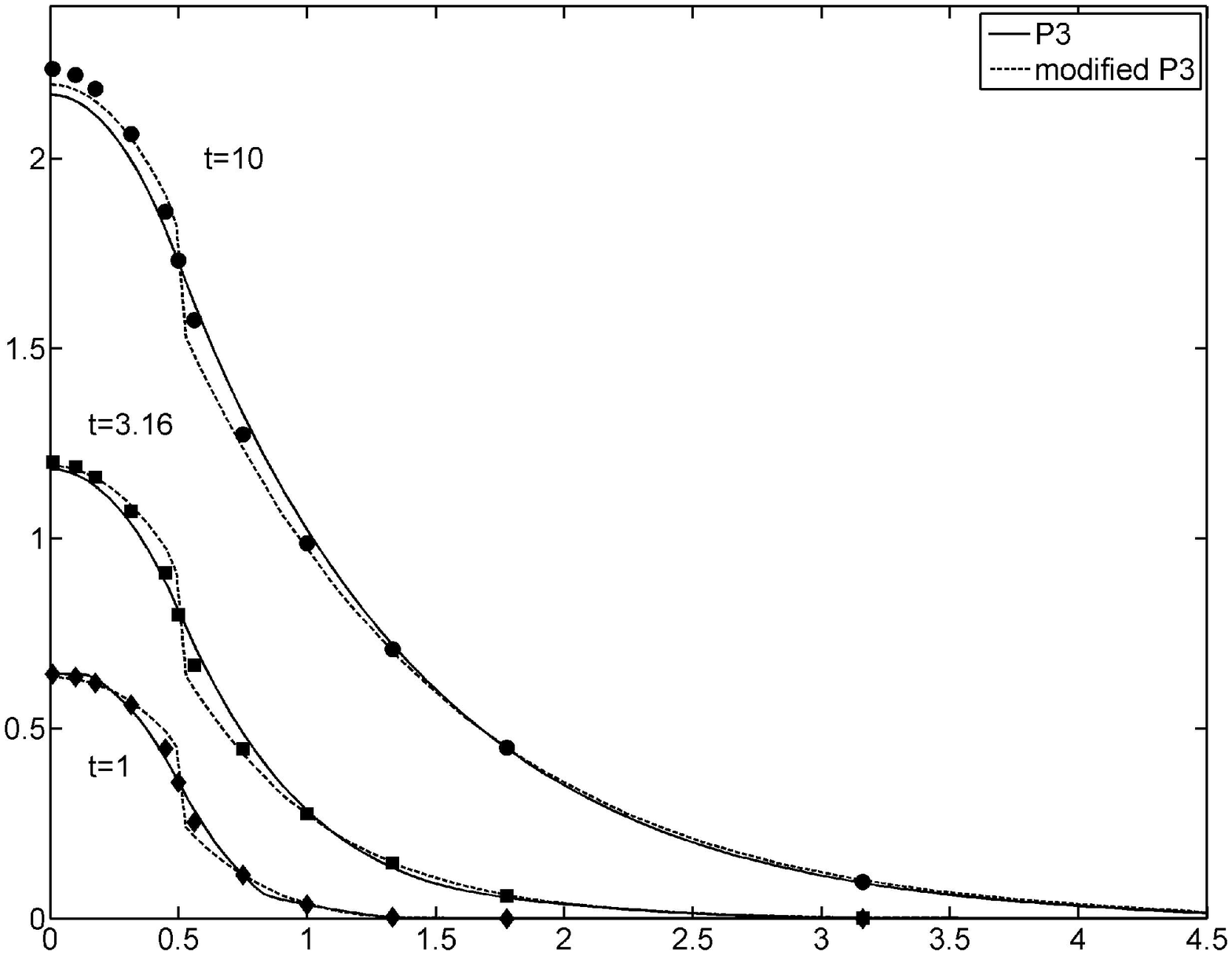}}
    \caption{Energy distribution at time $t=1\,s$, $t=3.16\,s$ and
             $t=10\,s$ for $P_N$ and $D_N$ approximations of
             different order.}
    \label{cha_numerics:fig:suolson_results}
\end{figure}
\begin{figure}[htbp]
    \begin{center}
        \includegraphicschoice{width=0.4\linewidth}{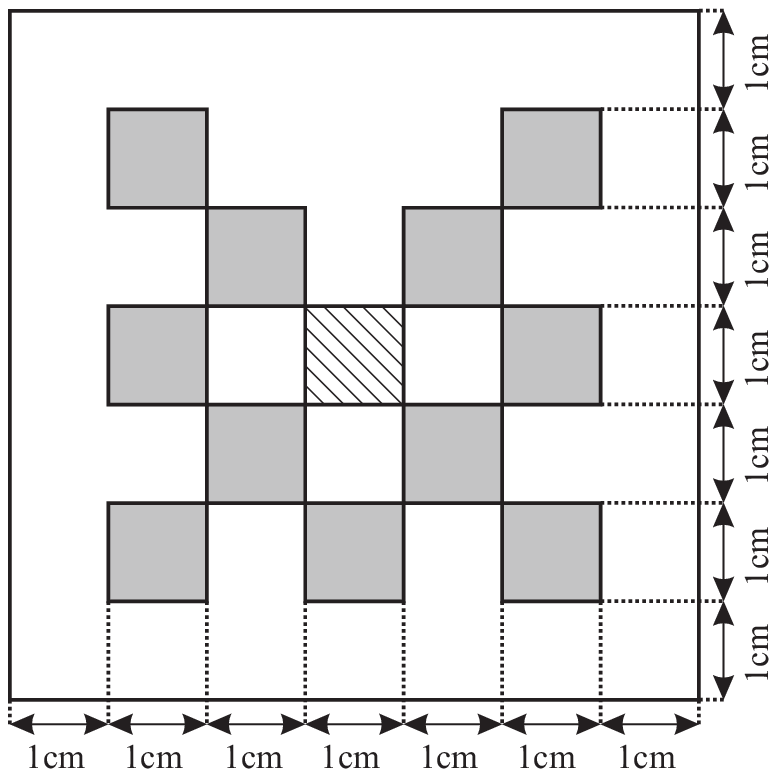}
    \end{center}
    \caption{Gray regions and the center area are highly absorbing while 
             white regions are highly scattering.  The radiation source is 
             located in the hatched center region.}
    \label{cha_numerics:fig:bruhol_geometry_setup}
\end{figure}

The thick black symbols mark the benchmark solution at times
$t=1\,s$, $t=3.16 \,s$ and $t=10\,s$.  The other curves are explained
in the legend of each plot.

   As we can see in \figref{cha_numerics:fig:pn_solution_suolson} and
\figref{cha_numerics:fig:mpn_solution_suolson}, the $P_N$ methods as
well as the $D_N$ approaches lead to solutions that
converge for increasing order $N$ to the benchmark solution.  But
comparing the order of the method that is necessary to reach a
specific accuracy shows that the $D_N$ approach leads to
similar results with less computational effort.  Additionally, we see
that for small times ($t=1$) both methods perform similarly well. But
for large times ($t=10$) the $D_1$ solution agrees already
very well with the benchmark solution while the $P_1$ solution is
much further away, especially in the region of the central peak.  
The solutions of order $5$ in the $D_N$ and $P_N$ approach 
(not shown) are almost identical and differ only in a few regions 
from the benchmark solution.

\subsection{Lattice Problem}
\label{cha_numerics:ssec:numerical_results_lattice_core}
This is an example with a complicated geometry, taken from 
\cite{Brunner2005}.  We consider a checkerboard of highly scattering 
and highly absorbing regions on a lattice core.  A graphical
representation of the setting is shown in
\figref{cha_numerics:fig:bruhol_geometry_setup}.
The white regions consist of a purely scattering material with
$\abc=0\, cm^{-1}$ and $\scc=1 \, cm^{-1}$. The eleven gray regions
and the central region are purely absorbing with $\abc=10\, cm^{-1}$
and $\scc=0 \, cm^{-1}$. For the propagation speed we assume $c=1 \,
cm/s$.  At time zero, a source of strength one is turned on in the
hatched central region.  The computational domain is surrounded on
all sides by vacuum boundaries.

The numerical results presented here have been obtained using a finite element
discretization with streamline diffusion. We used between 25000 and 400000 
bilinear elements. More details on the method can be found in \cite{Schaefer08}.

In \figref{cha_numerics:fig:bruhol_results} we present the energy
distribution ($\int_\sphere \I(\Omega) d\Omega$) of the radiative
field $3.2$ seconds after the radiation source in the center is turned
on.  The scale is logarithmic ($\log_{10}$).  We compare $P_N$ methods
with $D_N$ methods of different order.
\begin{figure}[htbp]
    \centering
    \subbottom[$P_1$]
        {\label{cha_numerics:fig:p1_solution_holloway}
        \includegraphicschoice{width=0.425\textwidth}{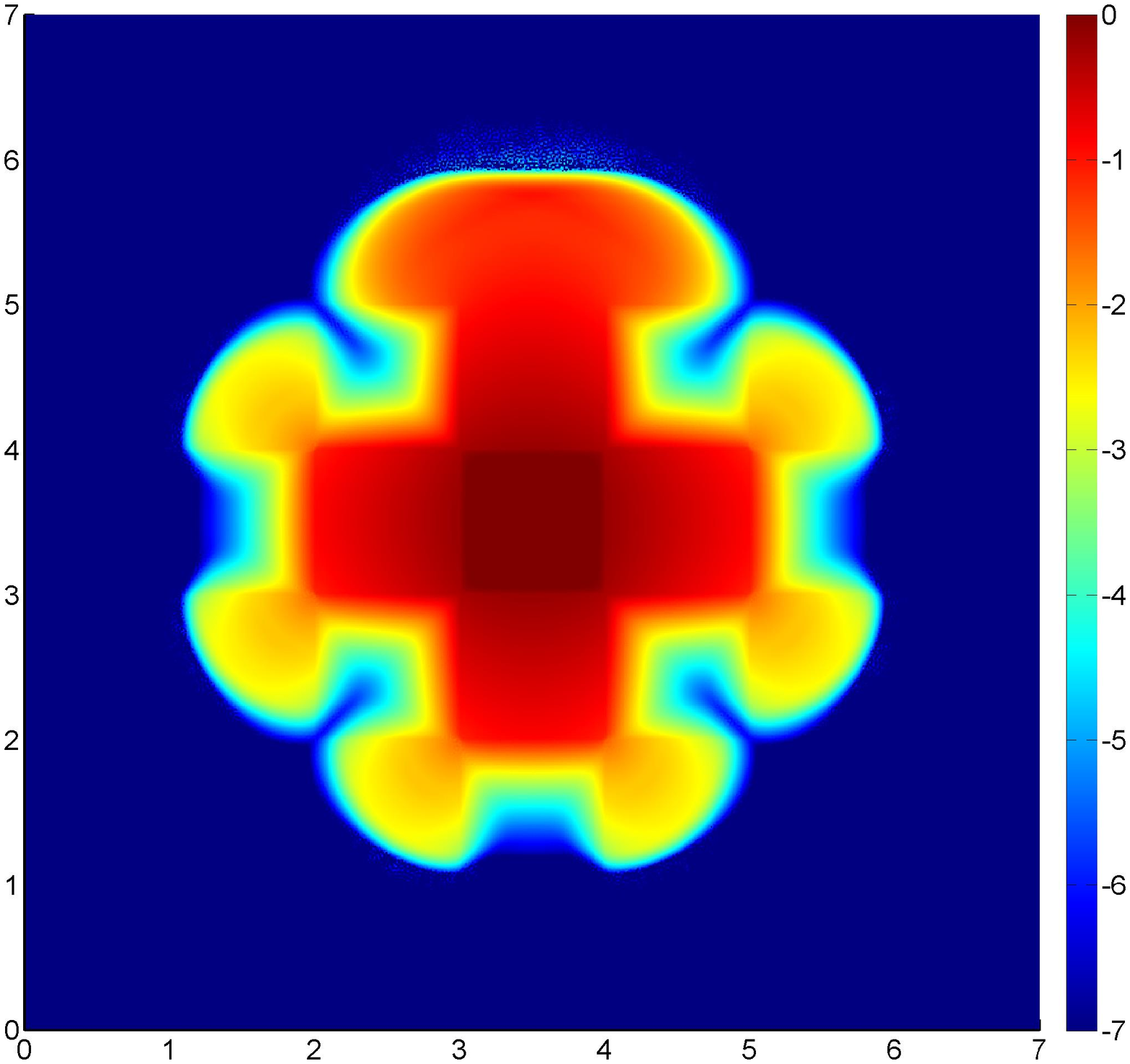}}\hspace{1cm}
    \subbottom[$D_1$]
        {\label{cha_numerics:fig:mp1_solution_holloway}
        \includegraphicschoice{width=0.425\textwidth}{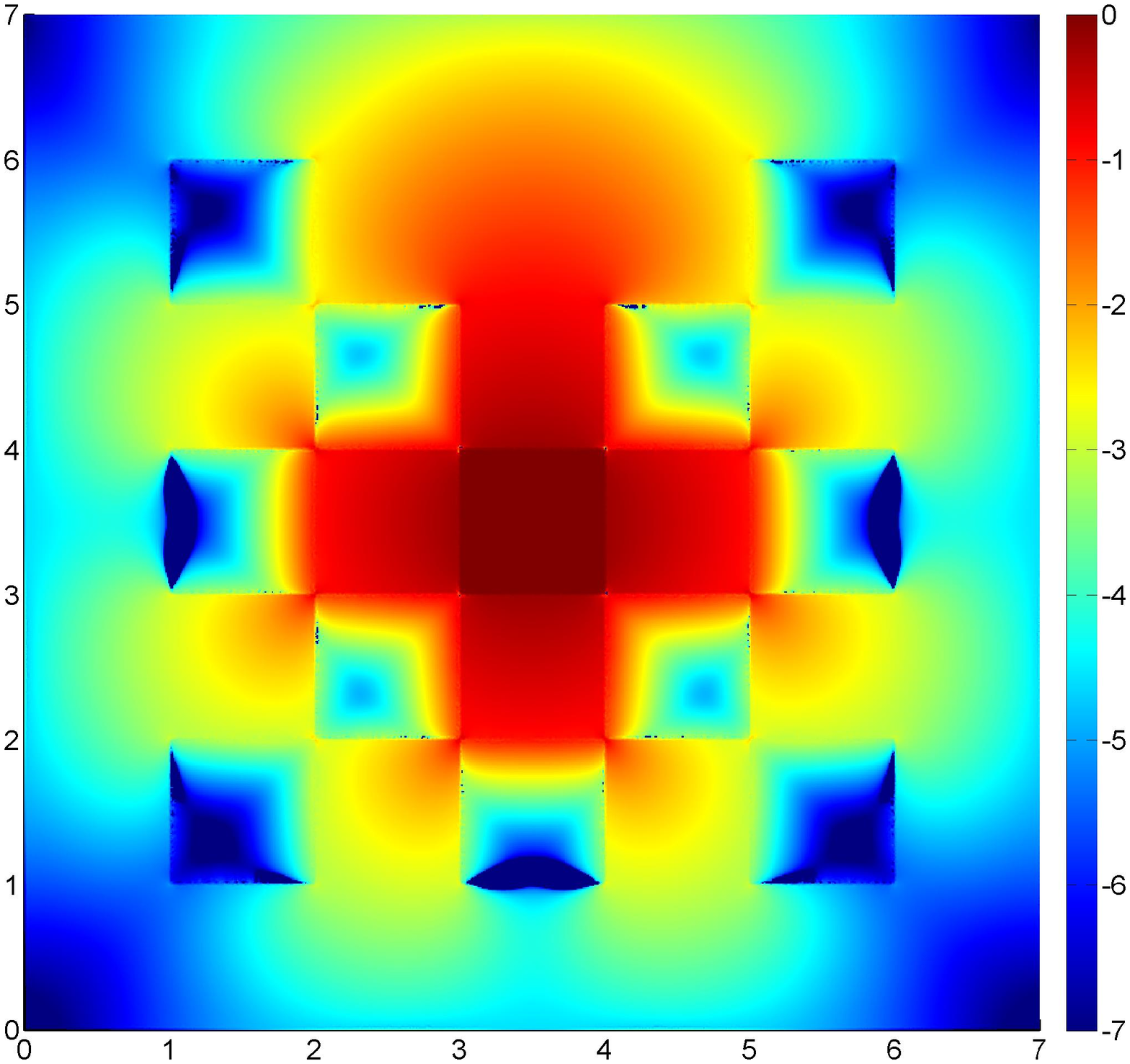}}
    \subbottom[$P_3$]
        {\label{cha_numerics:fig:p3_solution_holloway}
        \includegraphicschoice{width=0.425\textwidth}{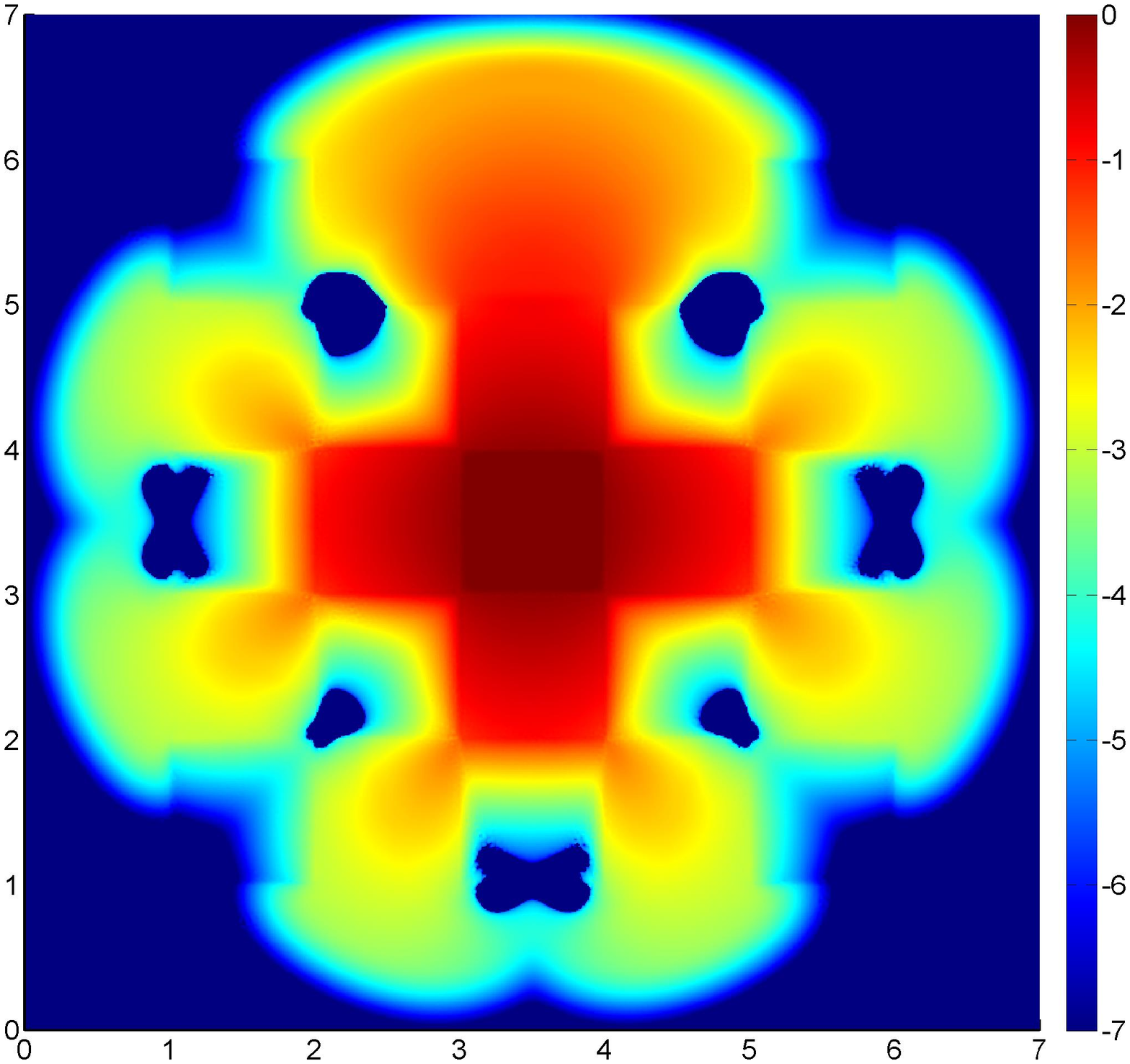}}\hspace{1cm}
    \subbottom[$D_3$]
        {\label{cha_numerics:fig:mp3_solution_holloway}
        \includegraphicschoice{width=0.425\textwidth}{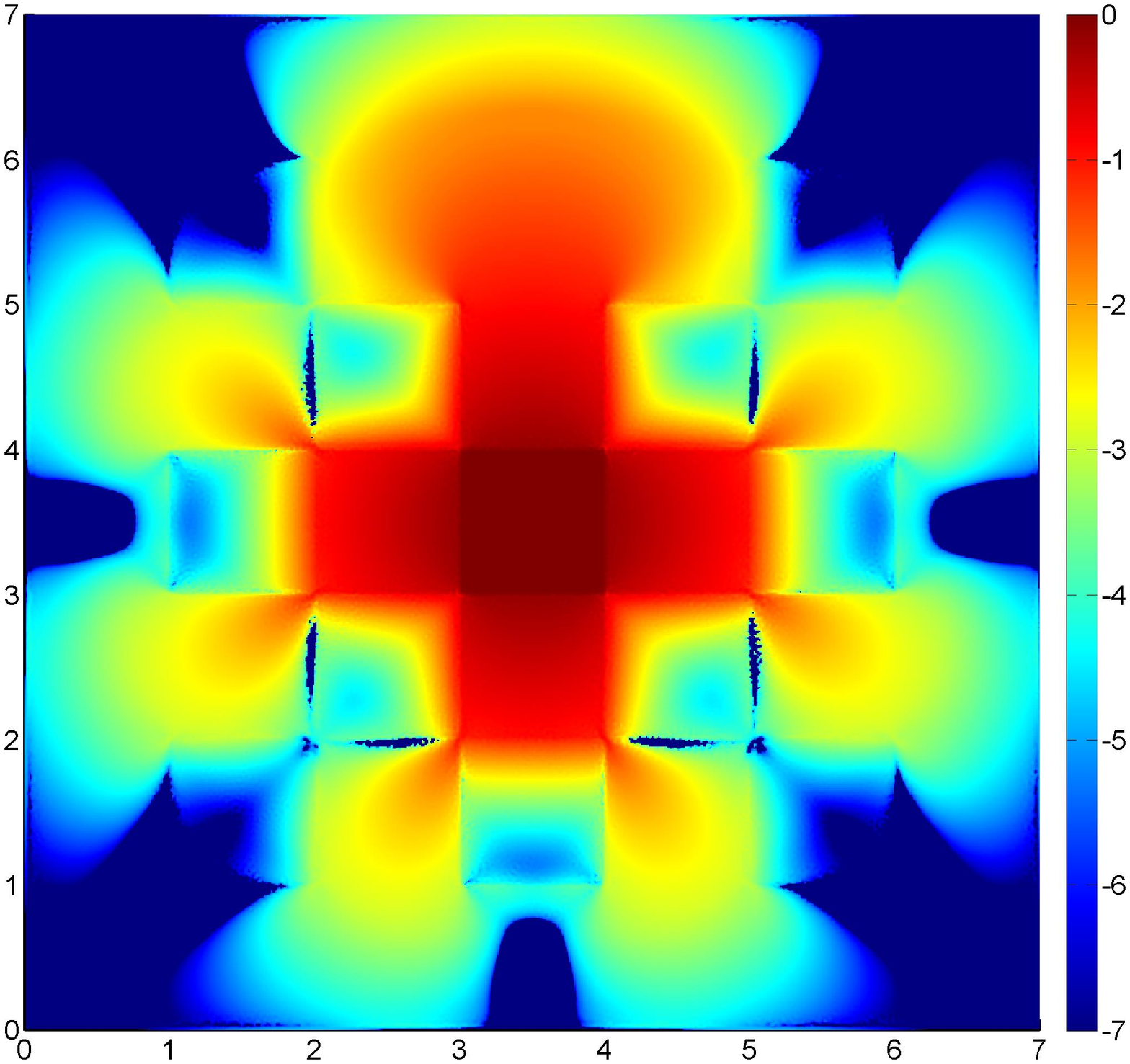}}
    \subbottom[$P_5$]
        {\label{cha_numerics:fig:p5_solution_holloway}
        \includegraphicschoice{width=0.425\textwidth}{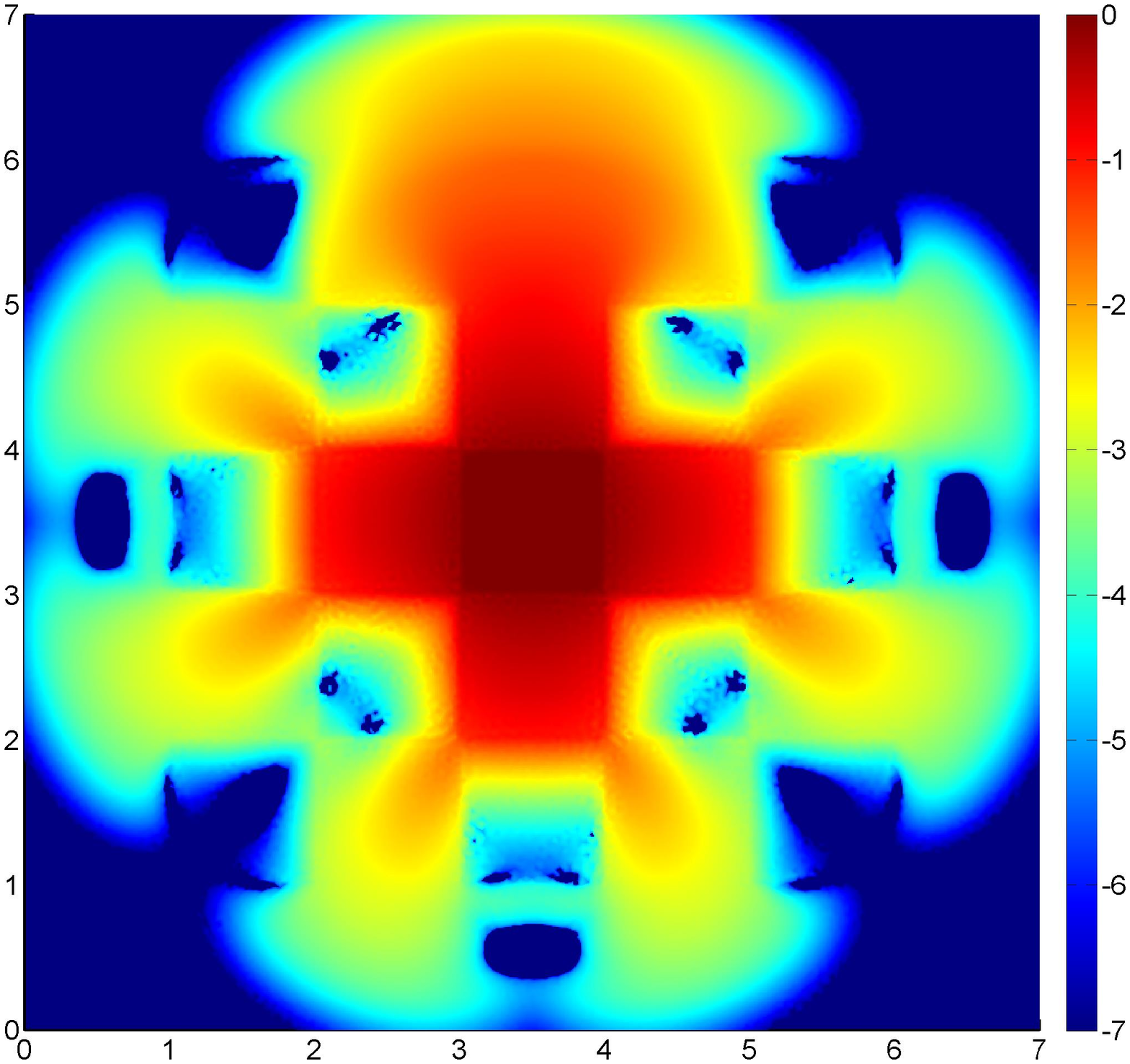}}\hspace{1cm}
    \subbottom[$D_5$]
        {\label{cha_numerics:fig:mp5_solution_holloway}
        \includegraphicschoice{width=0.425\textwidth}{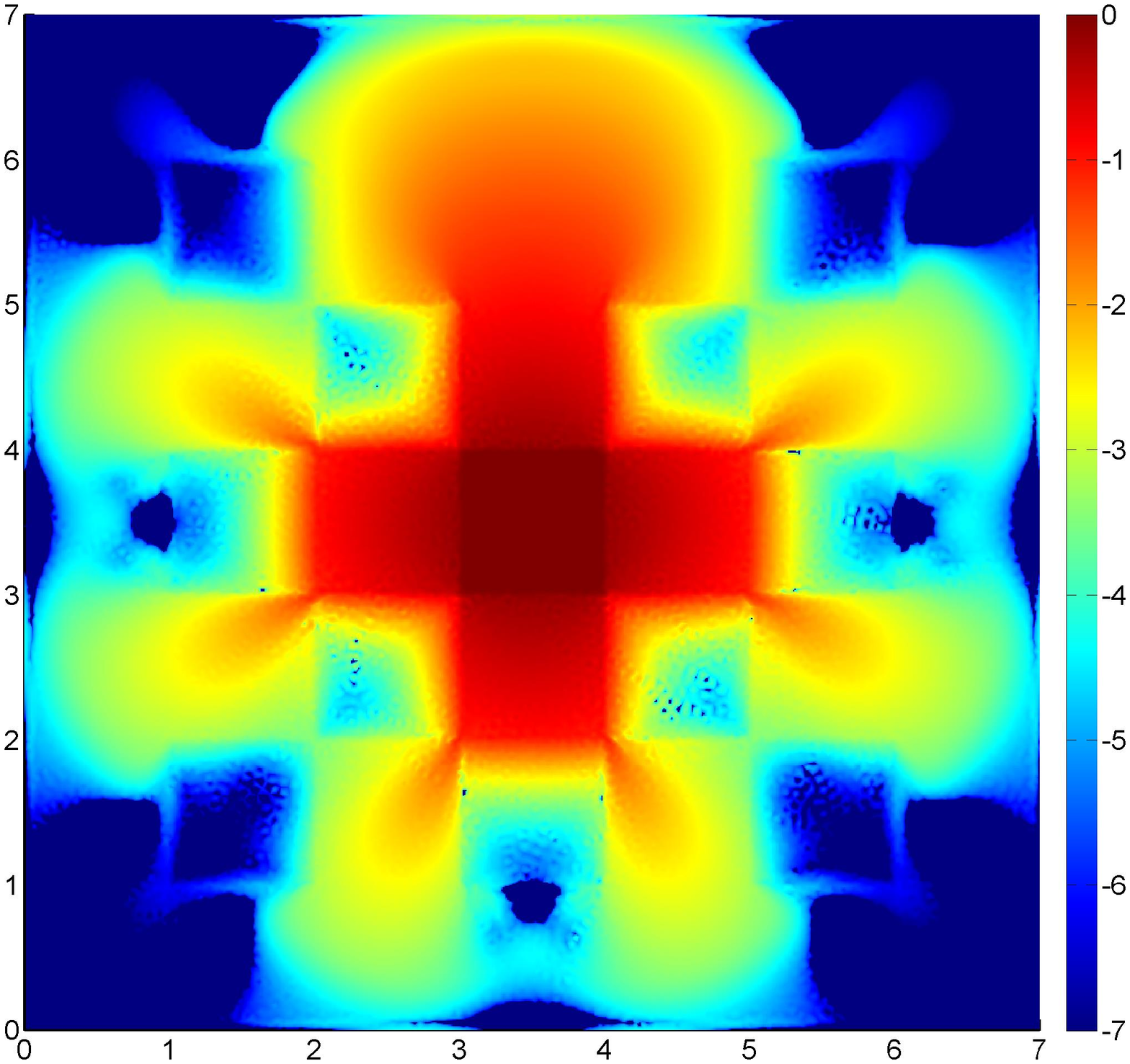}}
    \caption{Energy distribution for lattice problem approximated 
             by $P_N$ and $D_N$ methods of different order presented 
             in a logarithmic scale ($\log_{10}$).}
    \label{cha_numerics:fig:bruhol_results}
\end{figure}
\begin{figure}[htb]
    \centering
    \contsubbottom[$P_7$]
        {\label{cha_numerics:fig:p7_solution_holloway}
        \includegraphicschoice{width=0.425\textwidth}{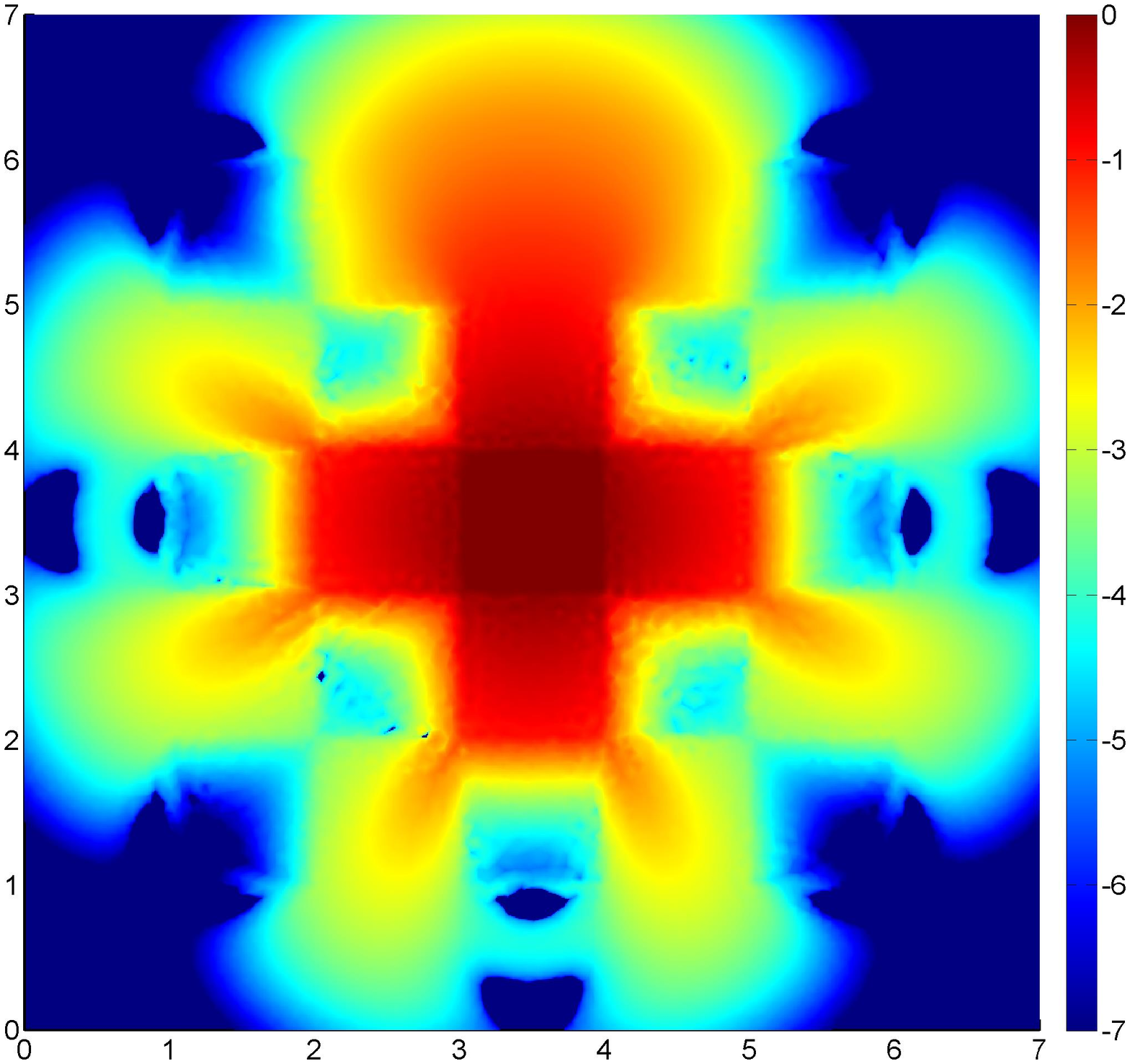}}\hspace{1cm}
    \contsubbottom[$D_7$]
        {\label{cha_numerics:fig:mp7_solution_holloway}
        \includegraphicschoice{width=0.425\textwidth}{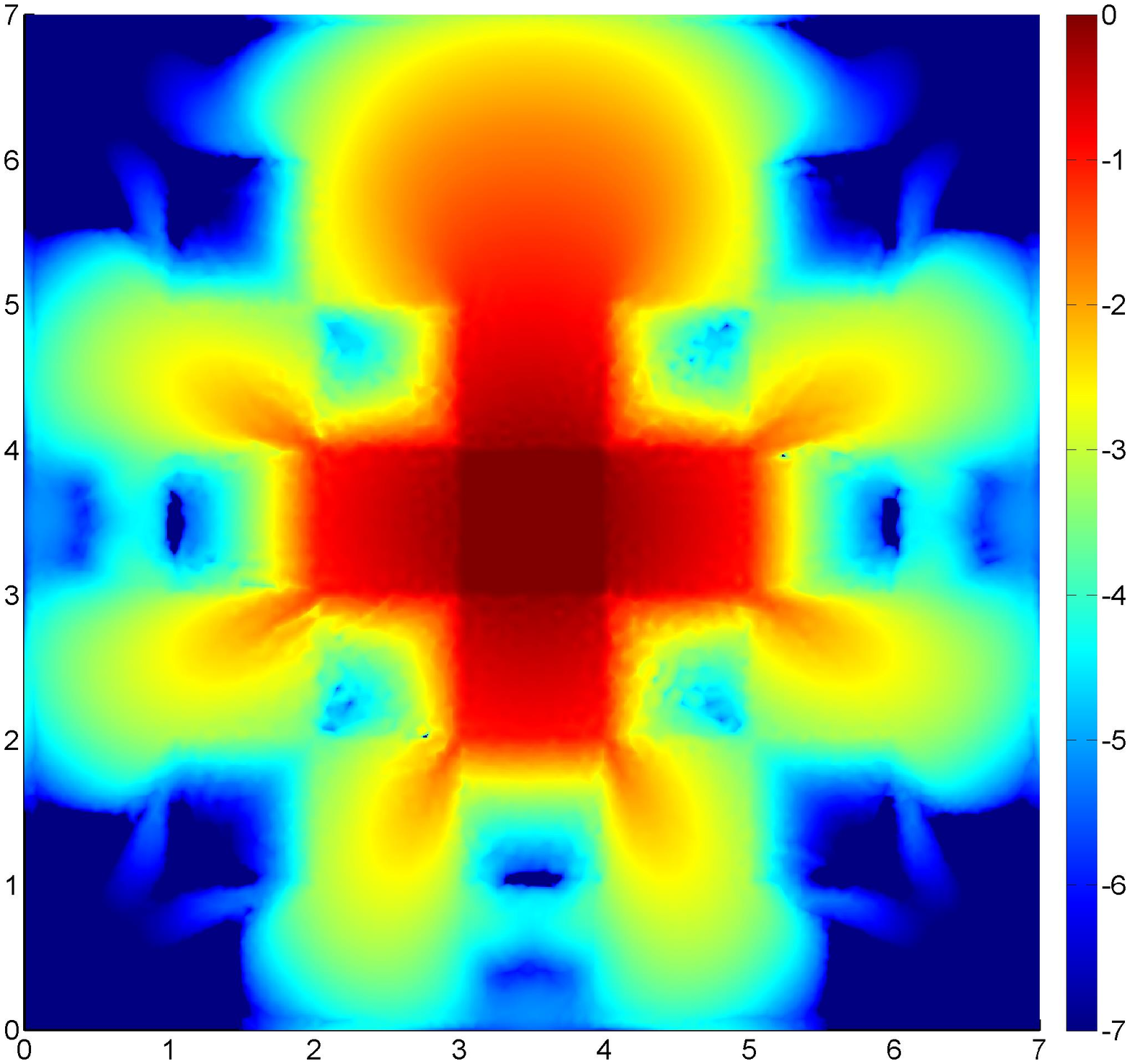}}
    \contcaption{Energy distribution for lattice problem approximated by $P_N$
                 and $D_N$ methods of different order presented in a logarithmic scale ($\log_{10}$). (continued)}
    \subconcluded
\end{figure}
The main differences in the solutions can be found in the beams
leaking between the corners of the absorbing regions, the shadows
behind the absorbing regions and the front of photons escaping from
the source region.

As we can see in the resulting figures, for increasing order, both
approaches converge to the same solution which for the $P_7$ and
$D_7$ models is almost the same as the one obtained by
Monte Carlo simulations in \cite{Brunner2005}.  But the $D_N$ model 
gives much better results for lower order approximations than the 
$P_N$ model.  In particular, the front of the escaping photons is 
tracked much better and the shadows behind the absorbing regions 
are more visible in lower order $D_N$ approaches.

The fact that the front of photons is not captured that well by the
$P_N$ method is related to the hyperbolic structure of the
equations.  Especially in the $P_1$ model the information can be
distributed only with one characteristic speed of $1/\sqrt{3}\,
cm/s$.  But that is far too slow, compared to the real speed of the
photons ($1 \, cm/s$).  The higher the order $N$ of the $P_N$
approximations the more the characteristic speed of the equations 
approaches the desired one and therefore the front can be tracked
much better (see \cite{Str98}).  Due to the additional diffusive term
introduced by the deviation approximation into the $D_N$
approximation, this effect is not present there.

\bibliographystyle{amsplain}
\bibliography{Lib_file}

\appendix
\section{Properties of Spherical Harmonics}
\label{apx:A}
Spherical harmonics are used as a set of basis functions for
representing functions mapping from the unit sphere into the complex
numbers. A good overview on the basics and properties of
spherical harmonics can be found in \cite{Arfken1970}. Here only the
most important properties related to moment methods will be
recalled.

Since Spherical harmonics act on the unit sphere, a parameterization is needed
\begin{equation}
    \sphere = \{\Omega \in \R^3: \; \Omega
              = \begin{pmatrix}
                  \sqrt{1-\mu^2} \cos \varphi \\
                  \sqrt{1-\mu^2} \sin \varphi \\
                  \mu \\
                \end{pmatrix}, \; \varphi\in [0,2\pi], \; \mu\in [-1,1]\} \;.
\end{equation}
Using these coordinates and the associated Legendre polynomials gives rise
to
\begin{defs}
    \label{apx_sph:eq:spherical_harmonics_definition}
    For all $n\in\N_0$ and $m\in \{-n,\ldots,n\}$ the function
    \begin{equation}
        \begin{split}
            \sph^m_n: \; &\sphere \to \C \\
            & (\varphi,\mu) \mapsto (-1)^m
            \sqrt{\frac{2n+1}{4\pi}\frac{(n-m)!}{(n+m)!}}P^{m}_{n}(\mu) e^{im\varphi}
        \end{split}
    \end{equation}
    is called a spherical harmonic function of order $n$ and degree $m$, where
    the associated Legendre polynomials of order $n$ and degree $m$ are defined from the Legendre polynomials $P_n$ as
    \begin{equation}
        P^m_n(\mu) = (1-\mu^2)^\frac{m}{2}\frac{d^m}{d\mu^m}P_n(\mu) \qquad n\in \N_0,\; m=0,\ldots,n\;.
    \end{equation}
\end{defs}

If it is clear from the context, the dependence on $\Omega$ will be
neglected and we write $\sph^m_n = \sph^m_n(\Omega)$. For indices
$m\notin \{-n,\ldots,n\}$ the spherical harmonics are identically
zero. For the set of all spherical harmonics we introduce
\begin{defs}
    \label{apx_sph:def:vector_of_spherical_harmonics}
    The vector of all spherical harmonics is given by
    \begin{equation}
        \vY(\Omega) = \left(\sph^0_0(\Omega), \sph^{-1}_{1}(\Omega), \sph^{0}_{1}(\Omega), \sph^{1}_{1}(\Omega),\ldots \right)^T \subset l^2\left( L^2(\sphere, \C)\right)\;.
    \end{equation}
    The spherical harmonics up to order $N$ can be represented by
    \begin{equation}
        \vY_N(\Omega) = \left(\sph^0_0(\Omega), \sph^{-1}_{1}(\Omega),
                              \sph^{0}_{1}(\Omega), \ldots, \sph^{N-1}_{N}(\Omega),
                              \sph^{N}_{N}(\Omega) \right)^T\;.
    \end{equation}
\end{defs}
Again, if it is clear what we are referring to we neglect the
dependence on $\Omega$ and simply write $\vY$ and $\vY_N$.

  We use the following properties of spherical harmonics to derive 
moment models.  There is a relation between spherical harmonics and 
their complex conjugated counterpart
\begin{equation}
    \label{apx_sph:eq:conjugated_spherical_harmonics}
    \sph^m_n(\Omega) = (-1)^m \overline{\sph^{-m}_n(\Omega)}\;,
\end{equation}
an addition theorem which leads to a relation between spherical
harmonics and Legendre polynomials
\begin{equation}
    \label{apx_sph:eq:addition_theorem_spherical_harmonics}
    P_n(\Omega\cdot \Omega') = \frac{4\pi}{2n+1} \sum_{k=-n}^n
    \sph^k_n(\Omega) \; \overline{\sph^k_n(\Omega')} \;,
\end{equation}
and the recursion relations
\begin{equation}
    \label{apx_sph:eq:recursion_relations_spherical_harmonics}
    \begin{split}
        e^{-i\varphi}\sin \theta \; \sph^{m}_{n} &= h_1(m,n) \sph^{m-1}_{n+1} - l_1(m,n) \sph^{m-1}_{n-1}\\
        e^{i\varphi}\sin \theta \; \sph^{m}_{n}  &=-h_2(m,n) \sph^{m+1}_{n+1} + l_2(m,n) \sph^{m+1}_{n-1}\\
        \cos \theta \; \sph^{m}_{n}              &= h_3(m,n) \sph^{m}_{n+1}   + l_3(m,n) \sph^{m}_{n-1}
    \end{split}
\end{equation}
with the coefficients
\begin{align}
    \label{apx_sph:eq:definition_of_recursion_coefficients_for_sph}
    h_1(m,n) = \sqrt{\frac{(n-m+1)(n-m+2)}{(2n+1)(2n+3)} } \quad & \quad l_1(m,n) = \sqrt{\frac{(n+m)(n+m-1)}{(2n-1)(2n+1)} } \notag\\
    h_2(m,n) = \sqrt{\frac{(n+m+1)(n+m+2)}{(2n+1)(2n+3)} } \quad & \quad l_2(m,n) = \sqrt{\frac{(n-m)(n-m-1)}{(2n-1)(2n+1)} }\\
    h_3(m,n) = \sqrt{\frac{(n-m+1)(n+m+1)}{(2n+1)(2n+3)} } \quad & \quad l_3(m,n) = \sqrt{\frac{(n-m)(n+m)}{(2n-1)(2n+1)} } \;. \notag
\end{align}
It is easy to see that for the coefficients in
\eqref{apx_sph:eq:definition_of_recursion_coefficients_for_sph} it holds
\begin{equation}
    \label{apx_sph:eq:sph_coefficient_relations}
    \begin{split}
    h_1(m,n) = h_2(-m,n)\;, \quad & \quad l_1(m,n)=l_2(-m,n)\;,\\
    h_3(m,n) = h_3(-m,n)\;, \quad & \quad l_3(m,n)=l_3(-m,n)\;.
    \end{split}
\end{equation}
Using the given recursion relations leads for any vector $\Omega$ on the unit
sphere to
\begin{equation}
    \label{apx_sph:eq:direction_and_sph}
    \Omega \sph^m_n=
    \begin{pmatrix}
      \frac{1}{2} \left( h_1(m,n) \sph^{m-1}_{n+1} - h_2(m,n) \sph^{m+1}_{n+1} - l_1(m,n) \sph^{m-1}_{n-1} + l_2(m,n) \sph^{m+1}_{n-1}\right)\\
      \frac{i}{2} \left( h_1(m,n) \sph^{m-1}_{n+1} + h_2(m,n) \sph^{m+1}_{n+1} - l_1(m,n) \sph^{m-1}_{n-1} - l_2(m,n) \sph^{m+1}_{n-1}\right) \\
      h_3(m,n) \sph^{m}_{n+1} + l_3(m,n) \sph^{m}_{n-1}
    \end{pmatrix}.
\end{equation}

To express the relations \eqref{apx_sph:eq:direction_and_sph} as
matrix--vector multiplications we introduce some new matrices. For
$j\in\{1,2,3\}$ the matrices $L_{x_j}^k \in \R^{2k+1} \times
\R^{2k-1}$ and $U_{x_j}^k \in \R^{2k-1} \times \R^{2k+1}$ are
defined as
\begin{subequations}
    \label{apx_sph:eq:system_submatrices_L}
    \begin{align}
        \label{apx_sph:eq:system_submatrices_L_x1}
        \left(L_{x_1}^k\right)_{(r,s)} &=
            \begin{cases}
                -l_1(r-k-1,k) & \text{for} \;  r=s+2 \\
                 l_2(r-k-1,k)  & \text{for} \; r=s \\
                 0      & \text{otherwise}
            \end{cases} \;, \\
        \label{apx_sph:eq:system_submatrices_L_x2}
        \left(L_{x_2}^k\right)_{(r,s)} &=
            \begin{cases}
                 -l_1(r-k-1,k) & \text{for} \; r=s+2 \\
                 -l_2(r-k-1,k) & \text{for} \; r=s \\
                 0     & \text{otherwise}
            \end{cases} \;, \\
        \label{apx_sph:eq:system_submatrices_L_x3}
        \left(L_{x_3}^k\right)_{(r,s)} &=
            \begin{cases}
                 l_3(r-k-1,k)   & \text{for} \; r=s+1 \\
                 0     & \text{otherwise}
            \end{cases}
    \end{align}
\end{subequations}
and
\begin{subequations}
    \label{apx_sph:eq:system_submatrices_H}
    \begin{align}
        \label{apx_sph:eq:system_submatrices_H_x1}
        \left(U_{x_1}^k\right)_{(r,s)} &=
            \begin{cases}
                 h_1(r-k,k-1) & \text{for} \; r=s \\
                -h_2(r-k,k-1) & \text{for} \; r=s-2 \\
                 0            & \text{otherwise}
            \end{cases} \;, \\
        \label{apx_sph:eq:system_submatrices_H_x2}
        \left(U_{x_2}^k\right)_{(r,s)} &=
            \begin{cases}
                h_1(r-k,k-1) & \text{for} \; r=s \\
                h_2(r-k,k-1) & \text{for} \; r=s-2 \\
                 0            & \text{otherwise}
            \end{cases} \;, \\
        \label{apx_sph:eq:system_submatrices_H_x3}
        \left(U_{x_3}^k\right)_{(r,s)} &=
            \begin{cases}
                 h_3(r-k,k-1) & \text{for} \; r=s-1 \\
                 0            & \text{otherwise}
            \end{cases} \;.
    \end{align}
\end{subequations}
We used the definitions of $l_i(\cdot,\cdot)$ and $h_i(\cdot,\cdot)$
from
\eqref{apx_sph:eq:definition_of_recursion_coefficients_for_sph}. The
resulting matrices have only two (or for $L_{x_3}^k$ and $U_{x_3}^k$
only one) diagonals that do not vanish. Their structure is
\begin{equation}
    \begin{split}
        & U_{x_1}^k =
            \begin{pmatrix}
              h_1 & 0 & -h_2 &  &  \\
                & h_1 & 0 & -h_2 &  \\
               &  & \ddots & \ddots & \ddots \\
            \end{pmatrix} \qquad
        U_{x_2}^k =
            \begin{pmatrix}
              h_1 & 0 & h_2 &  &  \\
                & h_1 & 0 & h_2 &  \\
               &  & \ddots & \ddots & \ddots \\
            \end{pmatrix} \\
        & \qquad \qquad \qquad \qquad \qquad
        U_{x_3}^k =
            \begin{pmatrix}
              0 & h_3 & 0 &  &  \\
                & 0 & h_3 & 0 &  \\
               &  & \ddots & \ddots & \ddots \\
            \end{pmatrix}
    \end{split}
\end{equation}
\begin{equation}
    L_{x_1}^k =
        \begin{pmatrix}
          l_2 &  &  \\
          0 & l_2 &  \\
          -l_1 & 0 & \ddots \\
           & -l_1 & \ddots \\
           &  & \ddots \\
        \end{pmatrix} \qquad
    L_{x_2}^k =
        \begin{pmatrix}
          -l_2 &  &  \\
          0 & -l_2 &  \\
          -l_1 & 0 & \ddots \\
           & -l_1 & \ddots \\
           &  & \ddots \\
        \end{pmatrix} \qquad
    L_{x_3}^k =
        \begin{pmatrix}
          0 &  &  \\
          l_3 & 0 &  \\
          0 & l_3 & \ddots \\
           & 0 & \ddots \\
           &  & \ddots \\
        \end{pmatrix}
\end{equation}
In this representation we neglected the dependence of the
coefficients $l_1, h_1, \ldots$ on the order of the spherical
harmonics. These relations can be found in
\eqref{apx_sph:eq:system_submatrices_L} and
\eqref{apx_sph:eq:system_submatrices_H}.

The matrices can be combined to larger matrices $\hat
L_{x_j}^N$ and $\hat U_{x_j}^N$ in the following way
\begin{equation}
    \label{apx_sph:eq:structure_complex_system_matrices_subparts}
    \hat L_{x_j}^N =
    \begin{pmatrix}
      0         &           &        &           &  \\
      L_{x_j}^1 & 0         &        &           &  \\
                & L_{x_j}^2 & \ddots &           &  \\
                &           &        & \ddots    &  \\
                &           &        & L_{x_j}^N & 0\\
    \end{pmatrix}
    \quad \text{and} \quad
    \hat U_{x_j}^N =
    \begin{pmatrix}
      0         & U_{x_j}^1 &           &        &  \\
                &           & U_{x_j}^2 &        &  \\
                &           & \ddots    &        &  \\
                &           &           & \ddots & U_{x_j}^N \\
                &           &           &        & 0 \\
    \end{pmatrix} \;.
\end{equation}
Additionally we introduce the factor
\begin{equation}
    \gamma_j =
    \begin{cases}
        \frac{1}{2} & \text{for} \quad j = 1 \\
        \frac{i}{2} & \text{for} \quad j = 2 \\
        1 & \text{for} \quad j = 3
    \end{cases}
\end{equation}
and the unit vector
\begin{equation}
    \R^{(N+1)^2} \ni \ve^{N}_{(m,n)} = (0,\ldots,0,1,0,\ldots,0)^T
\end{equation}
which is one only in the component that is related to the spherical
harmonic $\sph^m_n$ in the vector $\vY_N$ such that
$\left(\ve^{N}_{(m,n)} \right)^T \vY_N =\sph^m_n$. Then, for
$n\in\{0,\ldots,N\}$ and $m\in\{-n,\ldots,n\}$, we can write
\begin{equation}
    \label{apx_sph:eq:omega_times_sph_matrix_notation}
    \Omega_j \sph^m_n = \gamma_j \left( \left(\ve^{N}_{(m,n)}\right)^T \hat L_{x_j}^{N} \vY_N + \left(\ve^{N+1}_{(m,n)}\right)^T \hat U_{x_j}^{N+1} \vY_{N+1}\right) \;.
\end{equation}
\begin{lem}
    \label{apx_sph:lem:symmetrie_properties_submatrices_sph}
    For the matrices $\hat L_{x_j}^N$ and $\hat U_{x_j}^N$ hold the
    relations
    \begin{equation}
        \left(  \hat L_{x_1}^N  \right)^T = \hat U_{x_1}^N \;, \qquad
        \left(  \hat L_{x_2}^N  \right)^T = -\hat U_{x_2}^N \;, \qquad
        \left(  \hat L_{x_1}^N  \right)^T = \hat U_{x_1}^N \;.
    \end{equation}
\end{lem}
\section[Treatment of Scattering and Absorption]{Treatment of Scattering and Absorption operators in moment methods}
\label{apx:B}
The radiative transfer equation contains a scattering
component
\begin{equation}
    (\oS \I)(t,x,\Omega) 
    = \frac{\scc}{4\pi} \int_{S^2} \Phi(x,\Omega \cdot\Omega ') 
                                   \I(t,x,\Omega ') \, d\Omega'
\end{equation}
with a normalized scattering kernel $\Phi$.  This scattering kernel can be 
rather complicated.  Several theories have been developed to approximate 
realistic kernels. The first who established an approach was 
Mie \cite{Stratton1941}.  However, due to the complexity of his theory, 
several simplified approximations have been developed, e.g.\ the 
Henyey-Greenstein (HG) kernel \cite{Hulst1980} or the SAM 
(simplified approximate Mie) approach \cite{Liu94}.

To deal with the scattering kernel in moment methods, it is very
common to rewrite it into a series expansion based on Legendre
polynomials
\begin{equation}
    \Phi(x,\mu) = \sum_{j=0}^{\infty} \frac{2j+1}{2} \scc_j(x) P_j(\mu) \;.
\end{equation}
This leads to
\begin{equation}
    \oS I(\Omega)
    = \frac{\scc}{2} \sum_{l=0}^\infty
                     \sum_{k=-l}^l \scc_l \vI^k_l \sph^k_l(\Omega)
\end{equation}
and
\begin{equation}
    \mnt^r_s \oS \I(\Omega)
     = \frac{\scc}{2} \sum_{l=0}^\infty
                      \sum_{k=-l}^l \scc_l
                          \vI^k_l \left(\mnt^r_s\sph^k_l(\Omega)\right)
          = \frac{\scc}{2} \sum_{l=0}^\infty
                           \sum_{k=-l}^l \scc_l \vI^k_l \delta^k_r \delta^l_s
          = \frac{\scc}{2} \scc_s \vI^r_s \;.
\end{equation}
For the total absorption and scattering operator we get
\begin{equation}
        (\oK \I)(\Omega)
          = \sum_{l=0}^\infty\sum_{k=-l}^l \left(\abc+\scc - \frac{\scc}{2}\scc_l\right) \vI^k_l \sph^k_l(\Omega)
          =: \sum_{l=0}^\infty\sum_{k=-l}^l \tilde\scc_l \vI^k_l \sph^k_l(\Omega)
          \end{equation}
It is obvious that the inverse operator is
\begin{equation}
    (\oK^{-1} \I)(\Omega) = \sum_{l=0}^\infty\sum_{k=-l}^l \frac{1}{\abc+\scc - \frac{\scc}{2}\scc_l} \vI^k_l \sph^k_l(\Omega)\;.
\end{equation}
and thus
\begin{equation}
    \mnt^r_s(\oK^{-1} \I)(\Omega) = \frac{1}{\abc+\scc - \frac{\scc}{2}\scc_s} \vI^r_s.
    \label{apx_scat:eq:single_moment_absorb_and_scattering}
\end{equation}

\end{document}